\documentclass[hidelinks, 10pt, twocolumn]{IEEEtran}
\IEEEoverridecommandlockouts
\usepackage{graphicx}
\usepackage{textcomp}
\def\BibTeX{{\rm B\kern-.05em{\sc i\kern-.025em b}\kern-.08em
		T\kern-.1667em\lower.7ex\hbox{E}\kern-.125emX}}
\usepackage[normalem]{ulem}
\usepackage{amsmath}
\usepackage{amsthm,amssymb,amsmath,bm}
\usepackage{subfigure}
\usepackage{amsfonts}
\usepackage{epsfig}
\usepackage{amssymb}
\usepackage{amsmath}
\usepackage[table,xcdraw]{xcolor}
\usepackage[utf8x]{inputenc}
\usepackage{color,soul}
\usepackage{subfigure}
\usepackage{multirow}
\usepackage{rotating}
\usepackage{graphicx}
\usepackage{tabularx}
\usepackage{array}
\usepackage{color,soul}
\usepackage{bm}
\usepackage{graphicx}
\usepackage{blindtext}

\newtheorem{thm}{Theorem}

\usepackage{subfigure}
\usepackage{booktabs}
\usepackage{moreverb}
\usepackage{epsfig}
\usepackage{amsmath,amssymb,amsthm,mathrsfs,amsfonts,dsfont}
\usepackage{adjustbox,lipsum}
\usepackage{amsfonts}
\usepackage{epsfig}
\usepackage{amssymb}
\usepackage{amsmath}
\usepackage{amsthm}
\usepackage{subfigure}
\usepackage{multirow}
\usepackage{rotating}
\usepackage{graphicx}
\usepackage{tabularx}
\usepackage{array}
\usepackage{anyfontsize}
\usepackage{color,soul}
\usepackage{graphicx,dblfloatfix}
\usepackage{epstopdf}
\usepackage{blindtext}
\usepackage{amsmath}
\usepackage{amsthm,amssymb,amsmath,bm}
\usepackage{subfigure}
\usepackage{amsfonts}
\usepackage{epsfig}
\usepackage{amssymb}
\usepackage{amsfonts}
\usepackage{amsmath}
\usepackage{cite}
\usepackage{graphicx}
\usepackage{fancyhdr}
\usepackage{subfigure}
\usepackage[noblocks]{authblk}
\usepackage[subfigure]{tocloft}
\usepackage{subfigure}
\usepackage{tabularx}
\usepackage{tcolorbox}
\usepackage{cite}
\usepackage[linesnumbered,ruled,vlined]{algorithm2e}
\SetKwInput{KwInput}{Input}
\SetKwInput{KwOutput}{Output}
\usepackage{amsthm,amssymb,amsmath,bm}
\usepackage{graphicx}
\definecolor{azure(colorwheel)}{rgb}{0.0, 0.5, 1.0}
\usepackage[colorlinks=true,allcolors=azure(colorwheel)]{hyperref}%
\usepackage{fancyhdr}
\usepackage{subfigure}
\usepackage[subfigure]{tocloft}
\usepackage{subfigure}
\usepackage{tabularx}
\usepackage{tablefootnote}
\usepackage{cite}
\allowdisplaybreaks
\usepackage{hyperref}
\usepackage{lettrine}
\makeatletter
\def\endthebibliography{%
	\def\@noitemerr{\@latex@warning{Empty `thebibliography' environment}}%
	\endlist
}
\makeatother

\newcommand{\nosemic}{\renewcommand{\@endalgocfline}{\relax}}
\newcommand{\dosemic}{\renewcommand{\@endalgocfline}{\algocf@endline}}
\let\oldnl\nl
\newcommand{\nonl}{\renewcommand{\nl}{\let\nl\oldnl}}
\makeatother

\newcommand*{\myprime}{^{\prime}\mkern-1.2mu}

\begin{document}
	\title{AoI-Aware Resource Allocation for Platoon-Based C-V2X Networks via Multi-Agent Multi-Task Reinforcement Learning}
	\author{Mohammad Parvini, \textit{Student Member, IEEE},~Mohammad Reza Javan,~\textit{Senior Member, IEEE},~Nader Mokari,~\textit{Senior Member, IEEE},~Bijan Abbasi,~\textit{Senior Member, IEEE}, and  Eduard A. Jorswieck, \textit{Fellow, IEEE}
	\thanks{M. Parvini,~ N. Mokari, and B. Abbasi are with the Department of Electrical and Computer Engineering,~Tarbiat Modares University,~Tehran,~Iran, (e-mails: \{m.parvini, nader.mokari, abbasi\}@modares.ac.ir). M.~R.~Javan is with the Department of Electrical and Computer Engineering, Shahrood University of Technology, Iran, (e-mail: javan@shahroodut.ac.ir).~Eduard A. Jorswieck is with Institute for Communications Technology, TU
	Braunschweig, Germany, Email: jorswieck@ifn.ing.tu-bs.de.}}
	\maketitle

	\begin{abstract}
		
		This paper investigates the problem of age of information (AoI) aware radio resource management for a platooning system. Multiple autonomous platoons exploit the cellular wireless vehicle-to-everything (C-V2X) communication technology to disseminate the cooperative awareness messages (CAMs) to their followers while ensuring timely delivery of safety-critical messages to the Road-Side Unit (RSU). Due to the challenges of dynamic channel conditions, centralized resource management schemes that require global information are inefficient and lead to large signaling overheads. Hence, we exploit a distributed resource allocation framework based on multi-agent reinforcement learning (MARL), where each platoon leader (PL) acts as an agent and interacts with the environment to learn its optimal policy. Existing MARL algorithms consider a holistic reward function for the group's collective success, which often ends up with unsatisfactory results and cannot guarantee an optimal policy for each agent. Consequently, motivated by the existing literature in RL, we propose a novel MARL framework that trains two critics with the following goals: A global critic which estimates the global expected reward and motivates the agents toward a cooperating behavior and an exclusive local critic for each agent that estimates the local individual reward. Furthermore, based on the tasks each agent has to accomplish, the individual reward of each agent is decomposed into multiple sub-reward functions where task-wise value functions are learned separately. Numerical results indicate our proposed algorithm's effectiveness compared with the conventional RL methods applied in this area.
		\\
		\emph{\textbf{Index Terms---}} Resource management, V2X, AoI, Platoon cooperation,  MARL.
	\end{abstract}
	
	\section{Introduction}
	\lettrine[lraise=0, findent=0.1em, nindent=0.2em, slope=-.5em, lines=2]{I}{ntelligent} 
	 transportation systems (ITSs) will become a compulsory component of the future's smart cities. In essence, ITSs will address the issue of dense traffic networks and transportation bottlenecks by exploiting efficient traffic management approaches\cite{zhang2011data}. One of the foreseen services of ITS is the so-called autonomous vehicular platoon system\cite{hall2005vehicle}. Platooning is the first step toward fully autonomous driving, which is deemed one of the most representative potentials for overcoming the transport costs. Furthermore, platooning improves the intersection's operational efficiency compared to the case where cars cross the intersection one after another\cite{lioris2017platoons}.	 
	 In summary, a vehicle platoon is a convoy of interconnected vehicles that continuously coordinate their kinetics and share a typical moving pattern. In each platoon formation, the head-of-line vehicle is known as the Platoon Leader (PL), which is responsible for maintaining communication with other Platoon Members (PMs)\cite{jia2015survey}. In order to reap the benefits of the platooning system properly, several critical issues must be tackled. Firstly, each vehicle in the platoon must have enough awareness of its relative distance and velocity with its surrounding vehicles. This perception is needed to allow the vehicles in a platoon to regulate their decisions and to guarantee that any perturbation in the position or velocity of PL does not lead to amplified fluctuations in the behavior of PMs. This balance, known as the string stability, is ensured through the timely exchange of cooperative awareness messages (CAMs) among the vehicles of the platoon, and it is regularly initiated by the PL that manages the group\cite{ETSI_CAM_STANDARD}. 
	 Then,  every platoon must have sufficient information about the other existing platoons and vehicles in the network, especially in the case of intersections or road curves. 
	 These points reflect the importance of investigating an efficient resource allocation algorithm that meets the requirements of both inter-platoon and intra-platoon communications\cite{rios2016survey}.
	 
	 The advent of vehicle-to-everything (V2X) communication technology has addressed the aforementioned challenges. Platoons communicate with each other through the Road-Side Unit (RSU) with Vehicle-to-Infrastructure (V2I) communications in order to exchange the intersection safety messages, while vehicles in the same platoon exchange safety-critical messages by either broadcasting or cellular vehicle-to-vehicle (V2V) communications for CAM dissemination. The more frequently information is exchanged in the system, the sooner each platoon member can react and avoid prospective obstacles\cite{sensors_platoon_v2x}. 
	 The theoretical potential of Long Term Evolution (LTE) for V2X communications has been appraised in the Third Generation Partnership Project (3GPP) studies\cite{3gpp_36_885}. In LTE systems, eNodeBs centrally perform radio resource management (RRM). However, the conventional LTE architecture does not natively sustain direct V2V communications. Since LTE Release 12, 3GPP has provided several technical specifications to mitigate this problem through device-to-device (D2D) sidelink communications (also known as Proximity Services)\cite{3gpp_23_303,3gpp_36_785}. Furthermore, new requirements and use cases have been proposed for 5G V2X enhancements in Release 15 \cite{3gpp_22_886}. Following the existing literature, this paper is based on Mode 4, defined in the 3GPP cellular V2X architecture \cite{molina2017lte}. The resource scheduling and interference management between the platoons are established based on distributed algorithms implemented between the vehicles \cite{molina2017lte,LTE-V}.
	 
	 \subsection{Related Works}
	 Recently, the platooning system has been considered in various studies. 
	 The authors of \cite{Better_Platooning} analyze the capability of the LTE system in establishing intra-platoon communication.
	 In \cite{mei2018joint}, the authors study the reliability and efficiency of the platoon-based V2V communication, investigate the string stability requirements for the platooning systems and design a CAM dissemination mechanism in the LTE-V2V network. The authors of \cite{wang2019platoon} investigate the platoon cooperation in a multi-lane scenario and consider a two-step resource allocation along with developing a dynamic programming based subchannel allocation and power control algorithm to maximize the platoon size as well as to minimize the power consumption. In \cite{zeng2019joint}, string stability of the platoons and the maximum wireless system delay that guarantees the stability are analyzed. The resource allocation based on the evolved multimedia broadcast multicast services (eMBMS) capability and D2D communications is examined in \cite{peng2017resource} to enhance the reliability and reduce the transmission latency in a scenario with a chain of platoons. A two-stage platoon formation algorithm and a time division based intra-platoon resource allocation mechanism are introduced to develop stable platoons in \cite{wang2018resource}. 
	 Most of the issues that have been addressed in the articles mentioned above are related to the platoon's communications and interactions with each other or controlling algorithms employed to ensure the platoon's string stability. Nonetheless, an essential common concern that has not yet been elucidated is the fast-changing channel condition in vehicular environments that provoke uncertainty and inaccuracy in estimating the channel state information (CSI). On the other hand, the gradual increase in users' number leads to more complicated optimization problems with often nonlinear constraints, making them challenging to optimize by traditional optimization methods. The aforementioned hurdles call for investigating novel methods that can deal with more complex situations efficiently. 
	 
	 As one of the robust machine learning tools, reinforcement learning (RL) has recently attracted substantial attention. In \cite{zia2019distributed}, the authors analyze the spectrum allocation scheme by devising a distributed Q-learning approach, where autonomous D2D users try to maximize their throughput and minimize their interference to cellular users. Furthermore, an intelligent resource management problem in the Internet of Vehicles (IoV) networks is analyzed in \cite{yang2019intelligent} using an actor-critic RL method. However, the RL methods applied in the above works are suitable in low-dimensional state and action spaces. 
	 RL in combination with deep learning has led to the emergence of deep reinforcement learning (DRL)\cite{sun2019application}. DRL has sparked a flurry of interest and has found its way into vehicular network literatures\cite{V2X_Modeselection_DRL, Multi_Agent_DRL_Urban}. The authors of \cite{ye2019deep} propose a decentralized resource allocation method in a vehicular network for both unicast and broadcast scenarios employing DRL. In \cite{An_Intelligent_Path_Planning_DRL}, a mobile edge computing-based platooning system has been proposed in which the platoons locate their optimal path through RL. The authors of \cite{vu2020multi} investigate the problem of channel assignment and power allocation in a platooning vehicular network using the DRL approach. In a similar framework, spectrum and energy efficiency of vehicular platooning network is examined in \cite{liu2020joint}. In \cite{liang2019spectrum}, the authors investigate the spectrum sharing in a vehicular network by implementing a multi-agent DRL method. In order to tackle the problem of the environment's non-stationarity, the authors propose a fingerprint method that incorporates agents' policies in the observation space. Spectrum allocation for D2D communication is investigated in \cite{li2019multi} in which the authors propose a multi-agent actor-critic method. 
	 We can summarize the deficiencies of the works mentioned above as follows:  \cite{ye2019deep, An_Intelligent_Path_Planning_DRL} and \cite{liu2020joint} model the policy search as a Markov decision process (MDP), which means that all the agents update their policies independently. However, although these algorithms are capable of handling many complex problems, they cannot be applied to multi-agent systems (MASs). In MASs, all the agents act simultaneously and affect the environment, leading to a non-stationary environment \cite{feriani2021single}. On the other hand, \cite{vu2020multi} and \cite{liang2019spectrum} are based on multi-agent DRL.
	 DRL methods employ discrete action spaces which is not preferable in power control scenarios leading to poor results.
	 A widely applied MARL framework is multi-agent deep deterministic policy gradient (MADDPG) \cite{li2019multi}. MADDPG is based on centralized training and decentralized execution in which each agent collects the information of other agents during the training time and then executes actions independently based on its observation.  However, the Achilles heel of this method is that the critic's input grows linearly with the number of agents. Furthermore, although these algorithms reach an optimal solution, there is no explicit notion of coordination between the agents.
	 
	In vehicular networks, the traffic and intersection safety information is time-critical, and hence acquiring timely, and fresh traffic updates are of significant importance. Recently, an emerging new metric has been employed for capturing the timeliness of the information, namely the age of information (AoI)\cite{kaul2011AoI}. By definition, AoI is the time elapsed since the most recent received information update (from RSU point of view) was generated (at the corresponding platoon). Unlike traditional metrics such as delay, AoI only takes the information that delivers fresh updates to the RSU into account\cite{ultra_aoi_vehicular}. One of the recent works in this area is \cite{chen2020age} where the authors formulate an AoI-aware radio resource management problem in a Manhattan grid V2V network. 
	\subsection{Contribution}
	This work considers the AoI minimization problem in high mobility vehicular platooning system, consisting of multiple connected and autonomous vehicles where PLs attempt to access the frequency spectrum to disseminate the CAM messages between their followers while keeping an updated connection with the RSU. The novelty of this work lies in the following key contributions:
	\begin{itemize}
	\item We formulate a multi-objective optimization problem for each platoon to jointly minimize the AoI and maximize the CAM message transmission probability.
	\item We model the spectrum access of the multiple PLs as a multi-agent problem and exploit the recent progress of MARL structures in \cite{sheikh2020multi} to build a novel MARL framework on top of deterministic policy gradients architectures which trains two critics: A global critic which estimates the global expected reward and motivates collaboration between multiple agents, and an exclusive local critic for each agent that estimates the local expected reward.
	\item In order to tackle the problem of the overestimation bias in Q-functions, we exploit the Twin Delayed Deep Deterministic Policy Gradient (TD3) algorithm \cite{fujimoto2018addressing} for the global critic.
	\item Furthermore, by treating each sub-objective as a separate task, the individual reward of each agent is decomposed into multiple sub-reward functions where task-wise value functions are learned separately.
	\item Numerical experiments indicate that the proposed framework converges 3 times faster than the conventional RL frameworks and maintains the average AoI quantity within 5-10 milliseconds range, and guarantees a CAM message transmission probability of over 99 \% for various platoon sizes.
	\end{itemize}

	\begin{table}[!t]
		\centering
		\caption{Primary Notations used in the paper}
		\label{Table_Not}
		{\renewcommand{\arraystretch}{1.2}
			\begin{tabular}{{@{}lll@{}}}
				\toprule[1pt]
				\textbf{Notation} & \textbf{Definition}  \\
				\midrule
				$ \mathbb{N} $ & the set natural numbers\\
				$ P/\mathcal{P}/j $    & number/set/index of platoons  \\
				$ N_j/\mathcal{N}_j/n $ & number/set/index of vehicles in platoon $ j $\\
				$ K/\mathcal{K}/k $ & number/set/index of subchannels\\
				$ \alpha_j $ & frequency independent large-scale fading\\
				$ g_j[k] $ & frequency dependent small-scale fading \\
				$ \Re $ & RSU location \\
				$ \beta_{j,k}^t $ & subchannel allocation indicator \\
				$ \theta_{j}^t $ & inter/intra-platoon mode selection indicator \\
				$ \mathcal{C}_{j,\Re}^{t}[k] $ & data rate between PL $ j $ and the RSU in subchannel $ k $\\
				$ h_{j,\Re}[k] $ & channel gain from PL $ j $ to RSU in subchannel $ k $ \\
				$ \mathcal{C}_{j,i}^{t}[k] $ &  data rate between PL $ j $ and its follower $ i\in\mathcal{N}_j $ \\
				$ h_{j,i}[k] $ & channel gain from PL $ j $ to its PMs in subchannel $ k $ \\
				$ p^t_{j}[k] $ & power usage of PL $ j $ \\
				$ A_j^t $ & AoI of PL $ j $ up to the beginning of scheduling slot $ t $ \\
				$ \zeta_j $ & CAM messages size of PL $ j $\\
				$ \mathcal{C}_{j,\Re}^{\text{min}} $ & minimum capacity requirement of PL\\
				\bottomrule
		\end{tabular}}
	\end{table}
	\subsection{{Paper Organization and Notations}}
	The remainder of the paper is arranged as follows. 
	In Section \ref{S_M}, we discuss the proposed system model. Section \ref{MARLFRAMEWORK} describes the multi-agent reinforcement learning algorithm. In Section \ref{Result}, we present the simulation results and analyses, and finally, Section \ref{Conclusion} concludes the paper.
	
	\textit{Notations:}
	Most of the notations applied in this paper are standard. 
	To ease readability, all the primary notations of the paper are listed in Table \ref{Table_Not}.
	\section{System Model and Problem Formulation}\label{S_M}
	
	\begin{figure}[!t] 
		\centering
		\includegraphics[width=.5\textwidth]{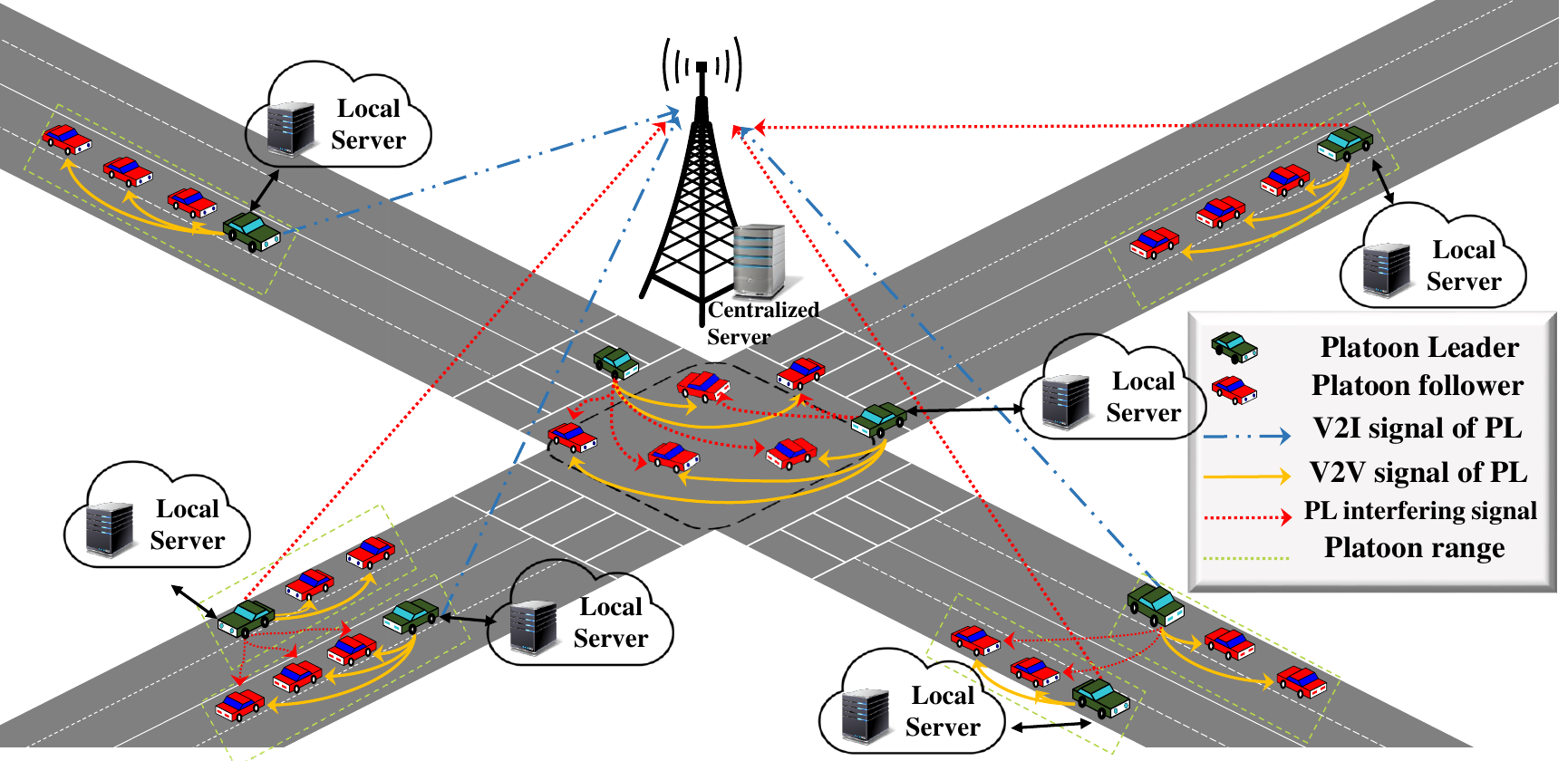}
		\caption{The multi-lane platoon scenario.} 
		\label{system_model}
	\end{figure}
	We consider a cellular V2X based vehicular communication network which consists of one RSU and multiple platoons, as shown in Fig. \ref{system_model}. The RSU is located at the center of the crossroad and is equipped with single antenna. We assume $ \mathcal{P} = \{1,2,\dots,P\},~P\in\mathbb{N} $, indicates the set of platoons. Each platoon itself is comprised of some connected and automated vehicles. Let $ \mathcal{N}_j = \{1,2,\dots,N_j\},~N_j\in\mathbb{N} $, be the number of vehicles in each platoon $ j \in \mathcal{P} $ which are numbered sequentially from one to $ N_j $, starting from PL. We discretize the time horizon into equal scheduling slots of length $ \Delta t $,
	indexed by a positive integer $ t \in \mathbb{N} $.  	
	The system bandwidth is divided into orthogonal subchannels of size $ W $. They are indexed by $ k \in \mathcal{K}=\{1, 2, \dots, K\} $. In essence, there are two types of communication modes in a platooning system, namely the intra-platoon and inter-platoon communication. In intra-platoon communication, vehicles within the same platoon, exchange the CAM information periodically through V2V links. According to the 3GPP specifications, \cite{3gpp_22_886}, CAMs dissemination frequency must be  between 10 to 100 Hz. In other words, the CAMs distribution period must be kept in the range of 100 ms or fewer. In inter-platoon communication, the RSU exchanges the intersection safety and platoon control information with every platoon via the V2I links. The first one is crucial in terms of guaranteeing the platoon string stability which lets the vehicles keep a close distance with each other and ensuring that all the platoon members are aware of the kinematics and the decisions of the other platoon members, especially the platoon leader. The latter is essential to inform all the platoons to become aware of the other platoons' status and traffic condition of the intersection.
	We exploit the orthogonal frequency division multiplexing (OFDM) to cope with the frequency selective wireless channels\footnote{It is necessary to mention that in this work, we only consider the channel gains related to the platoon leader interactions with the RSU and its followers.}. Furthermore, we assume that the channel fading is independent across different subchannels and remains constant within one sub channel. We model the channel gain of PL $ j \in \mathcal{P} $ in subchannel $ k $ during one coherence time period $ t $ as
	\begin{equation}
	h_j^t[k] = \alpha_j^t g_j^t[k],
	\label{inter_platoon_channel}
	\end{equation}
	where $ \alpha_j^t $ and $ g_j^t[k] $ denote the large-scale fading effect comprised of path loss and shadowing, and small-scale fading, respectively. Moreover, we define the binary variable $ \beta^t_{j,k}\in\{0, 1\} $ that indicates whether subchannel $ k $ is allocated to platoon $ j $ at time slot $ t $. Then PL $ j $ will decide whether to use the allocated subchannel for inter-platoon (i.e., to communicate with the RSU) or intra-platoon (i.e., to broadcast the CAM to its followers) communication. For this reason, we define another binary decision variable $ \theta^{t}_{j}\in\{0, 1\} $ that indicates the platoon leader's decision. When $ \theta^{t}_{j} = 1 $, that means that the PL will utilize the allocated subchannel for broadcasting (intra-platoon) and $ \theta^{t}_{j} = 0 $ indicates that the subchannel will be used for V2I (inter-platoon) communication.	
	We can express the instantaneous rates achieved in V2I communications between PL $ j $ and the RSU according to the Shannon capacity formula as follows:
	\begin{equation}
	\begin{aligned}
	&\mathcal{C}_{j,\Re}^{t}[k]=\log_2
	\left(1+\frac{(1-\theta^t_{j})\beta^t_{j,k} p^t_{j}[k]h^{t}_{j,\Re}[k]}{{I}^{t}_j[k]+\sigma^{2}}\right),\\
	&{I}^{t}_j[k] = \resizebox{.4\hsize}{!}{$ \sum_{j^{\prime}}\beta^t_{j^{\prime},k}p^t_{j^{\prime}}[k]{h}^{t}_{j^{\prime},\Re}[k]$},~~j\neq j^{\myprime},
	\end{aligned}
	\end{equation}
	where the interference from other platoons is treated as noise, $ p^t_{j}[k] $ is the transmit power level used by PL $ j $ on subchannel $ k $,  $ h^{t}_{j,\Re}[k] $ is the channel gain from PL $ j $ to RSU in subchannel $ k $, $ \sigma^2 $ is the noise power, $ \Re $ indicates the RSU location, $ {h}^{t}_{j^{\myprime},\Re} $ is the interfering channel to the RSU from PL $ j^{\myprime} \in \mathcal{P} $  functioning in whether inter ($ \theta_{j}^t=0 $) or intra-platoon ($ \theta_{j}^t=1 $) communication mode, and $ {I}^{t}_j[k] $ represents the total interference power. 
	Furthermore, we can calculate the instantaneous rates between PL $ j $ and its follower $ i $ as
	\begin{equation}
	\begin{aligned}
	&\mathcal{C}_{j,i}^{t}[k] =\log\left(1+\frac{\theta^t_{j}\beta^t_{j,k} p^t_{j}[k]h^{t}_{j,i}[k]}{{I}_j^{\prime, t}[k]+\sigma^{2}}\right),\\
	&\resizebox{.9\hsize}{!}{${I}_j^{\prime, t}[k] =  \sum_{j^{\prime}}\beta^t_{j^{\prime},k}p^t_{j^{\prime}}[k]{h}^{t}_{j^{\prime},i}[k],~~j\neq j^{\myprime},~~i\in\mathcal{N}_j\backslash\{1\} $},
	\end{aligned}
	\end{equation}
	where $ p_j^t[k] $ is the power used by PL $ j $, $ h^{t}_{j,i}[k] $ is the channel gain from PL $ j $ to its PMs in subchannel $ k $, $ {h}^{t}_{j^{\myprime},i} $ is  the interfering channel to PL $ j $'s members from PL $ j^{\myprime} \in \mathcal{P} $  functioning in whether inter ($ \theta_{j}^t=0 $) or intra-platoon ($ \theta_{j}^t=1 $) communication mode, and $ {I}_j^{\prime,t}[k] $ represents the total interference power. As described earlier, the PL has to maintain timely communication with the RSU to exchange the intersection safety messages. In this regard, we note $ A^{t}_{j} $ as the AoI of platoon $ j\in\mathcal{P}$ up to the beginning of scheduling slot $ t $, that is, the time elapsed since the most recently successful V2I communication\cite{kaul2011AoI}. The AoI of platoon $ j\in\mathcal{P}$ evolves according to
	\begin{align}
	\resizebox{.91\hsize}{!}{$A_{j}^{t+1}=
		\begin{cases}
		\Delta t, & \text { if }  (1-\theta^t_{j})\beta^t_{j,k}\cdot \mathcal{C}_{j,\Re}^{t}[k]\geq\mathcal{C}_{j,\Re}^{\text{min}}, \\
		A^{t}_{j}+\Delta t, & \text { otherwise }
		\end{cases}$}
	\label{AoI_evolution}
	\end{align}
	where $ \mathcal{C}_{j,\Re}^{\text{min}} $ is the minimum capacity requirement of V2I communication. 
	As (\ref{AoI_evolution}) suggests, within every successful transmission between the RSU and PL $ j\in\mathcal{P} $, the AoI will reset to $ \Delta t $. 
	Accordingly, we can express the multi-objective optimization problem (MoP) for platoon $ j $ as
	\begin{equation}
	\begin{aligned}
	\min_{\boldsymbol{\beta},\boldsymbol{\theta},\boldsymbol{p}}& \Bigg\{\frac{1}{T}\sum_{t=1}^{T}A^{t}_{j}, -\text{ Pr }\bigg\{\sum_{t=1}^{T}\sum_{k\in\mathcal{K}}\min_i\left\{\mathcal{C}_{j,i}^{t}[k]\right\}\Delta t\geq\zeta_j\bigg\},\\ &\frac{1}{T}\sum_{t=1}^{T}\sum_{k\in\mathcal{K}}p^t_j[k]\Bigg\},\\
	\textbf{s.t.} \quad 
	&C1: \mathcal{C}_{j,\Re}^{t}[k] \geq \mathcal{C}_{j,\Re}^{\text{min}},\quad\quad~\forall j\in\mathcal{P},~\forall k\in\mathcal{K},\\
	&C2: \beta_{j,k}^t\,,\theta^t_{j}\in\{0, 1\},~~~~\forall j\in\mathcal{P},~\forall k\in\mathcal{K},\\
	&C3: \sum_{k\in\mathcal{K}}\beta_{j,k}^t \leq 1,\quad\quad~~~\forall j\in\mathcal{P},~\forall t\in\mathbb{N},\\
	&C4:  p_j^t[k] \leq p_j^{\text{max}},\quad\quad\quad\forall j\in\mathcal{P},~\forall k\in\mathcal{K},\\
	\end{aligned}
	\label{optimization_problem}
	\end{equation}
	where $ \zeta_j $ is the CAM message size. 
	The objective is to minimize the expected AoI and power consumption  for every platoon while maximizing the probability of CAM messages delivery rate among the PMs within every $ T $ seconds\footnote{
		As stated in Section \ref{S_M}, $ T $ must be below 100 ms according to \cite{3gpp_22_886}.}. Constraint C3 shows that each platoon can access only one subchannel in every time slot and constraint C4 is to satisfy that the transmit power of PL $ j $ remains below its maximum value $ p_j^{\text{max}} $. 
	In optimization problem (\ref{optimization_problem}), the mode selection indicator $ \theta^t_j $ and subchannel selection indicator $ \beta $ are both binary variables. Furthermore, the objective function is non-convex. Consequently, the optimization problem (\ref{optimization_problem}) is a NP-hard combinatorial optimization problem\cite{V2X_Modeselection_DRL}, which is difficult to be solved. In this regard, we will investigate the state-of-the-art multi-agent deep deterministic policy gradient methods to address the complexities of the proposed optimization problem.
	\section{Multi-Agent RL Based Resource Allocation}\label{MARLFRAMEWORK}
	In this section, we will elaborate on the multi-agent environment and its associated states, actions, and rewards, and finally, we will discuss the proposed MARL algorithm and its relevant formulations. 
	\subsection{Modeling of Multi-Agent Environment}
	For a MARL with $ P $ agents (platoons), the optimization problems can be expressed as
	\begin{equation}\label{marlopt}
		\max_{\pi_{j}}\mathcal{J}_j(\pi_j),\quad j\in\mathcal{P},~~ \pi_j\in\Pi_j,
	\end{equation}
	where $ \mathcal{J}_j(\pi_j) = \mathbb{E}[\sum_{t=0}^{\infty}\gamma^tR^{t+1}_j|s^0_j] $,  $ \pi_j $ is the policy of agent $ j $, and $ \Pi_j $ is the set of all feasible policies for agent $ j $. 
	Each PL as an agent interacts with the vehicular network environment and takes action according to its policy, aiming at solving the optimization problem (\ref{optimization_problem}), or in other words, maximizing its total expected reward (\ref{marlopt}). At each time $ t $, the PL observes a state, $ s^t $, and accordingly takes action, $ a^t $. The environment transitions to a new state $ s^{t+1} $ and PL receives a reward based on its selected action. In our proposed system model the state space $ \mathcal{S} $, action space $ \mathcal{A} $, and the reward function $ r^t $, are defined as follows:
	 
	 $\bullet$ ~\textbf{State space:} 
	The state observed by the PL $ j $ (agent $ j $) at time slot $ t $ consists of several parts: the instant channel information between PL $ j $ and the RSU, $ h^t_{j,\Re}[k] $, for all $ k\in\mathcal{K} $, the channel information between PL $ j $ and its followers, $ h^t_{j,i}[k] $, $ i\in\mathcal{N}_j\backslash\{1\} $, the previous interference from other platoons to PL $ j $, $ I_j^{t-1}[k] $, the AoI of PL $ j $, $ A_{j}^{t} $, the remaining intra-platoon payload (CAM message) designated to be transferred by $ T $, $ \zeta^{r}_j $, and the remaining time budget, $ T^{r}_j $. Hence, the state space of PL $ j $ is
	$$ 
	\mathbf{s}^t_j = \left[h^t_{j,\Re}[k], h^t_{j,i}[k], I_j^{t-1}[k], A_{j}^{t},  \zeta^{r}_j, T^{r}_j\right],~~ j\in\mathcal{P}.
	$$
	
	$\bullet$ ~\textbf{Action space:}
	The action of each PL $ j \in \mathcal{P} $ is defined as $\mathbf{a}_j^t = \{\beta^t_j, \theta^t_j, p^t_j\}$. As mentioned earlier, $ \beta^t_j $ indicates which subchannel the PL $ j \in \mathcal{P} $ has selected, $ \theta^t_j $ represents the mode selection, and $ p^t_j $ represents the power control. It is noteworthy to mention that because we have applied the deep deterministic policy gradient method, the agent can select any power ranging from 0 to $ p_j^{\text{max}} $. This is the advantage of policy gradient methods that apply continuous actions spaces and can converge to more accurate results than conventional DQNs in which the power has to be discretized.
	
	$\bullet$ ~\textbf{Reward function:}
	What makes the reinforcement learning framework fascinating is the flexibility we have in designing the reward function that drives the learning process. In our proposed learning problem, the agents receive two reward signals, a global team reward, which evaluates the agents' cooperation, and an individual reward, which measures each agent’s performance. Accordingly, we first discuss the proposed learning algorithm and then return to the reward function's design. 
	
	The MARL frameworks' architecture is shown in Fig. \ref{MARL}, which is built on top of the MADDPG structure. In particular, we have designed two MARL frameworks, namely the \textit{Modified MADDPG} and \textit{Modified MADDPG with task decomposition}, in which the latter is the extension of the first one, where the holistic local reward function of each agent is further decomposed into sub-reward functions and learned separately. Unlike MADDPG, which uses a single critic to train multiple agents, the proposed framework trains two critics with the following functionalities: The centralized global critic, which is shared between all the agents, takes the observations and actions of all the agents as input and estimates the global team reward for them. The local critic, which is specific for each agent, receives the agent's local observation and action and estimates the local expected reward. In a sense, the goal is to simultaneously move the policy toward maximizing both global and local rewards and solve the optimization problem (\ref{optimization_problem}) for each agent. Furthermore, the agents do not necessarily need to know each other's policies and take actions based on their own observations. The agents' performance will be considered as ``decent" only when they act in a way that results in a proper global team reward as well as a satisfactory individual reward for each agent.
	\vspace*{-0.1em}
	\subsection{Modified MADDPG}
	Let \resizebox{0.18\textwidth}{!}{$ \Theta_{\pi}=\left(W_{\pi}^{(1)},\dots,W_{\pi}^{(L_{\pi})}\right) $} and \resizebox{0.18\textwidth}{!}{$ \Phi_{q}=\left(W_{q}^{(1)},\dots,W_{q}^{(L_{q})}\right) $}, be the parameter space of agents' actor and critic networks and \resizebox{0.18\textwidth}{!}{$ \Psi_{g}=\left(W_{g}^{(1)},\dots,W_{g}^{(L_{g})}\right) $} be the parameter space of the global critic, where $ L_{\pi}, L_{q}$ and $ L_{g} $ are the number of hidden layers in agents' actor and critic networks and the global critic, respectively. $ Ws$ are the neural networks' weight matrices and their dimensions are related to the number of nodes in the hidden layers. We consider a vehicular environment consisting of $ P $ platoons (agents) 
	with policies $ \boldsymbol{\pi} = \{\pi_1,\dots,\pi_P\} $. The agents' policies $ \pi_j $ and Q-functions $ Q^j_{\phi_j} $, and the global critic's Q-function $ Q^g_{\psi} $ are parameterized by $ \theta_{j} $, $ \phi_j $ and $ \psi $, respectively, where $ \theta_{j}\in\Theta_{\pi} $, $ \phi_j\in\Phi_{q} $ and $ \psi\in\Psi_{g} $. The MADDPG  for platoon $ j $ can be written as
	\begin{equation}
	\nabla_{\theta_j}\mathcal{J}_j=\mathbb{E}\left[\left.\nabla_{\theta_j} \pi_{j}\left(a_{j} \mid s_{j}\right) \nabla_{a_{j}} Q_{j}^{\boldsymbol{\pi}}\left(\mathbf{s}, \mathbf{a} \right)\right|_{a_{j}=\pi_{j}\left(s_j\right)}\right],
	\nonumber
	\end{equation}
	where $ \mathbf{s} = (s_1,\dots,s_P) $ and $ \mathbf{a} = (a_1,\dots,a_P) $ are the total state and action spaces. $ Q_{j}^{\pi}\left(\mathbf{s}, \mathbf{a}\right) $ is the centralized action-value function that takes the actions and states of the agents as its input to estimate Q-value for platoon $ j $. Based on the framework depicted in Fig. \ref{MARL}, the modified policy gradient for each agent $ j $ can be written as
	\vspace*{-0.5em}
	\begin{figure*}[!t] 
		\centering
		\includegraphics[width=.99\textwidth]{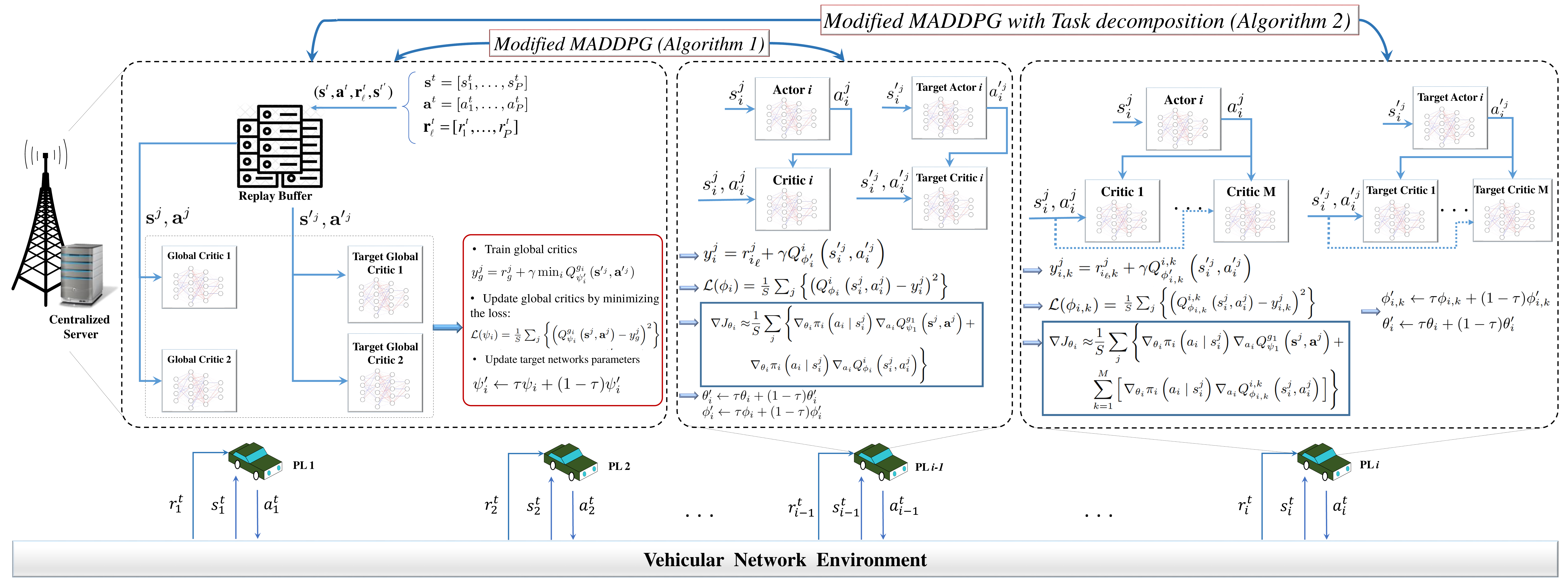}
		\caption{The architecture of the modified MADDPG and the modified MADDPG with task decomposition frameworks. The functionality of the global critic which is implemented at the RSU is similar for both the algorithms. However, they apply different procedures for the individual performance of the agents. (Notice the differences between the algorithms highlighted in blue boxes.)} 
		\label{MARL}
	\end{figure*}
	\begin{equation}\label{dec_maddpg}
	\begin{aligned}
	\nabla_{\theta_j}\mathcal{J}_j=&\underbrace{\mathbb{E}_{\mathbf{s}, \mathbf{a}\sim\mathcal{D}}\left[\left.\nabla_{\theta_{j}} \pi_{j}\left(a_{j} \mid s_{j}\right) \nabla_{a_{j}} Q^{g}_{\psi}\left(\mathbf{s}, \mathbf{a} \right)\right.\right]}_{\textit{Global Critic}} + \\
	& \underbrace{\mathbb{E}_{{s}_j, {a}_j\sim\mathcal{D}}\left[\left.\nabla_{\theta_{j}} \pi_{j}\left(a_{j} \mid s_{j}\right) \nabla_{a_{j}} Q^{j}_{\phi_j}\left({s}_j, {a}_j \right)\right.\right]}_{\textit{Local Critic}},
	\end{aligned}
	\end{equation}
	where $ a^t_j = \pi_j(s^t_j) $ is the action the agent $ j $ chooses following its policy $ \pi_j $. The first term in (\ref{dec_maddpg}) refers to the global critic which takes as input the agents' states and actions and estimates the team reward. The second term in (\ref{dec_maddpg}) refers to each agent's local critic that unlike the global critic, only takes each agent's local state and action to estimate the agent's individual performance. The global critic is updated as
	\begin{equation}
	\mathcal{L}(\psi)=\mathbb{E}_{\mathbf{s}, \mathbf{a}, \mathbf{r}, \mathbf{s}^{\prime}}\left[\left(Q^{g}_{\psi}\left(\mathbf{s}, \mathbf{a}\right)-y_{g}\right)^{2}\right],
	\end{equation}
	where $ y_g $ is the target value and is defined as follows:
	\begin{equation}
	y_{g}=r_{g}+\left.\gamma Q^{g}_{\psi^{\prime}}\left(\mathbf{s}^{\prime}, \mathbf{a}^{\prime}\right)\right|_{a_{j}^{\prime}=\pi_{j}^{\prime}\left(s_{j}^{\prime}\right)},
	\end{equation}
	where $ \boldsymbol{\pi}^{\prime} = \{\pi_1^{\prime},\dots,\pi_P^{\prime}\} $ refers to the target policies which are parameterized by $ \boldsymbol{\theta}^{\prime} = \{\theta_1^{\prime},\dots,\theta_P^{\prime}\} $. Similarly the local critic of agent $ j $, $ Q^j $, is updated by
	\begin{equation}\label{local_q}
	\mathcal{L}^j(\phi_j)=\mathbb{E}_{\mathbf{s}_j, \mathbf{a}_j, \mathbf{r}_j, \mathbf{s}^{\prime}_j}\left[\left(Q^{j}_{\phi_j}\left({s}_j, {a}_j\right)-y_{\ell}^j\right)^{2}\right],
	\end{equation}
	and $ y_{\ell}^j $ is defined as
	\begin{equation}\label{local_a}
	y_{\ell}^j=r_{\ell}^j+\left.\gamma Q^{j}_{\phi_j^{\prime}}\left({s}_j^{\prime}, {a}_j^{\prime}\right)\right|_{a_{j}^{\prime}=\pi_{j}^{\prime}\left(s_{j}^{\prime}\right)}.
	\end{equation}	
	Although the proposed framework can lead to decent results, there is still the problem of overestimation and suboptimal policies in Q-functions due to the function approximation errors. Motivated from the results in \cite{fujimoto2018addressing}, the global critic is replaced with the Twin delayed Deterministic Policy Gradient in  (\ref{dec_maddpg}). The resulting policy gradient is
	\begin{equation}\label{td3_dec_maddpg}
	\begin{aligned}
	\nabla_{\theta_j}\mathcal{J}_j=&\underbrace{\mathbb{E}_{\mathbf{s}, \mathbf{a}\sim\mathcal{D}}\left[\left.\nabla_{\theta_{j}} \pi_{j}\left(a_{j} \mid s_{j}\right) \nabla_{a_{j}} Q^{g_{1}}_{\psi_1}\left(\mathbf{s}, \mathbf{a} \right)\right.\right]}_{\textit{TD3 Global Critic}} + \\
	& \underbrace{\mathbb{E}_{{s}_j, {a}_j\sim\mathcal{D}}\left[\left.\nabla_{\theta_{j}} \pi_{j}\left(a_{j} \mid s_{j}\right) \nabla_{a_{j}} Q^{j}_{\phi_j}\left({s}_j, {a}_j \right)\right.\right]}_{\textit{Local Critic}}.
	\end{aligned}
	\end{equation}
	In (\ref{td3_dec_maddpg}), the twin global critics are updated as
	\begin{equation}
		\mathcal{L}(\psi_{i})=\mathbb{E}_{\mathbf{s}, \mathbf{a}, \mathbf{r}, \mathbf{s}^{\prime}}\left[\left(Q^{g_{i}}_{\psi_i}\left(\mathbf{s}, \mathbf{a}\right)-y_{g}\right)^{2}\right],
	\end{equation}
	where $ y_g $ is defined as follows:
	\vspace{-1em}
	\begin{equation}
	y_{g}=r_{g}+\left.\gamma \min_{i=1, 2}Q^{g_i}_{\psi_i^{\prime}}\left(\mathbf{s}^{\prime}, \mathbf{a}^{\prime}\right)\right|_{a_{j}^{\prime}=\pi_{j}^{\prime}\left(s_{j}^{\prime}\right)},
	\end{equation}
		\begin{algorithm}[!t]
		\small
		\caption{Modified MADDPG\label{algorithm1}}
		\renewcommand{\arraystretch}{0.5}{
			Start environment simulator  and generate platoons\\
			Initialize main global critic networks $ Q^{g_1}_{\psi_1} $ and $ Q^{g_2}_{\psi_2}$\\
			Initialize target global critic networks $ Q^{g_1}_{\psi_1^{\prime}} $ and $ Q^{g_2}_{\psi_2^{\prime}}$\\
			Initialize each agent's policy and critic networks
			
			\For{each episode}
			{Update platoons locations and respective channel gains\\
				Reset the Intra-platoon payload $ \zeta $ and maximum delivery time $ T $ to 100 ms 
				\\
				\For{each timestep t}
				{\For{each agent k}
					{
						Observe $ s_k^t $ and
						select action $ a^t_k = \pi_{\theta_{k}}(s_k^t) $\\
					}
					$ \mathbf{s}^t = [s_1^t,\dots,s_P^t] $,\quad
					$ \mathbf{a}^t = [a_1^t,\dots,a_P^t] $\\
					Receive global and local rewards, $ r_g^t $ and $ \mathbf{r}_l^t $\\
					Store $ (\mathbf{s}^t, \mathbf{a}^t, \mathbf{r}_l^t, r_g^t, \mathbf{s}^{t+1}) $ in replay buffer $ \mathcal{D} $
				}
				
				Sample minibatch of size S, $ (\mathbf{s}^j, \mathbf{a}^j, \mathbf{r}_g^j, \mathbf{r}_{\ell}^j, \mathbf{s}^{\prime_{j}}) $, from replay buffer $ \mathcal{D} $\\
				Set $ y_{g}^j=r_{g}^j+\gamma \min_{i}Q_{\psi^{\prime}_i}^{g_i}\left(\mathbf{s}^{\prime_{j}}, \mathbf{a}^{\prime_{j}}\right) $\\
				Update global critics by minimizing the loss: \\$ \mathcal{L}(\psi_{i})=\frac{1}{S}\sum_{j}\left\{\left(Q_{\psi_{i}}^{g_{i}}\left(\mathbf{s}^j, \mathbf{a}^j\right)-y_{g}^j\right)^{2}\right\} $\\
				Update target parameters:
				$ \psi_{i}^{\prime} \leftarrow \tau\psi_{i} + (1-\tau)\psi_{i}^{\prime} $\\
				\If{episode mod d}
				{
					Train local critics and actors\\
					\For{each agent i}
					{
						Set $ y^j_i=r_{i_{\ell}}^{j}+\gamma Q_{\phi_i^{\prime}}^{i}\left({s}^{{\prime}_j}_i, {a}^{{\prime}_j}_i\right) $\\
						Update local critics by minimizing the loss: \\$ \mathcal{L}(\phi_{i})=\frac{1}{S}\sum_{j}\left\{\left(Q_{\phi_i}^{i}\left({s}_i^j, {a}_i^j\right)-y^j_i\right)^{2}\right\} $\\
						Update local actors: \\
						$
						\begin{aligned}
						\nabla J_{\theta_{i}}\approx
						&\frac{1}{S}\sum_{j}\Bigg\{\nabla_{\theta_{i}} \pi_{i}\left(a_{i} \mid s_{i}^j\right) \nabla_{a_{i}} Q_{\psi_{1}}^{g_{1}}\left(\mathbf{s}^j, \mathbf{a}^j \right) + \\
						& {\nabla_{\theta_{i}} \pi_{i}\left(a_{i} \mid s^{j}_i\right) \nabla_{a_{i}} Q^{i}_{\phi_{i}}\left({s}_i^j, {a}_i^j \right)}\Bigg\}
						\end{aligned}
						$
						
						Update target networks parameters: \\
						$ \theta_{i}^{\prime} \leftarrow \tau\theta_{i} + (1-\tau)\theta_{i}^{\prime} $\\
						$ \phi_{i}^{\prime} \leftarrow \tau\phi_{i} + (1-\tau)\phi_{i}^{\prime} $
				}}
			}
		}
	\end{algorithm}
	\begin{algorithm}[!t]
		\small
		\caption{Modified MADDPG with TDec.\label{algorithm2}}
		\renewcommand{\arraystretch}{0.5}{
			Start environment simulator  and generate platoons\\
			Initialize main global critic networks $ Q^{g_1}_{\psi_1} $ and $ Q^{g_2}_{\psi_2}$\\
			Initialize target global critic networks $ Q^{g_1}_{\psi_1^{\prime}} $ and $ Q^{g_2}_{\psi_2^{\prime}}$\\
			Initialize each agent's policy networks\\
			Initialize each agent's task specific critic networks
			
			\For{each episode}
			{ Update platoons locations and respective channel gains\\
				Reset the Intra-platoon payload $ \zeta $ and maximum delivery time $ T $ to 100 ms 
				\\
				\For{each timestep t}
				{\For{each agent k}
					{
						Observe $ s_k^t $ and
						select action $ a^t_k = \pi_{\theta_{k}}(s_k^t) $\\
					}
					$ \mathbf{s}^t = [s_1^t,\dots,s_P^t] $,\quad
					$ \mathbf{a}^t = [a_1^t,\dots,a_P^t] $\\
					Receive global and local rewards, $ r_g^t $ and $ \mathbf{r}_l^t $\\
					Store $ (\mathbf{s}^t, \mathbf{a}^t, \mathbf{r}_l^t, r_g^t, \mathbf{s}^{t+1}) $ in replay buffer $ \mathcal{D} $
				}
				
				Sample minibatch of size S, $ (\mathbf{s}^j, \mathbf{a}^j, \mathbf{r}_g^j, \mathbf{r}_{\ell}^j, \mathbf{s}^{\prime_{j}}) $, from replay buffer $ \mathcal{D} $\\
				Set $ y_{g}^j=r_{g}^j+\gamma \min_{i}Q_{\psi^{\prime}_i}^{g_i}\left(\mathbf{s}^{\prime_{j}}, \mathbf{a}^{\prime_{j}}\right) $\\
				Update global critics by minimizing the loss: \\$ \mathcal{L}(\psi_i)=\frac{1}{S}\sum_{j}\left\{\left(Q_{\psi_{i}}^{g_{i}}\left(\mathbf{s}^j, \mathbf{a}^j\right)-y_{g}^j\right)^{2}\right\} $\\
				Update target parameters:
				$ \psi_{i}^{\prime} \leftarrow \tau\psi_{i} + (1-\tau)\psi_{i}^{\prime} $\\
				\If{episode mod d}
				{
					Train local critics and actors\\
					\For{each agent i}
					{\For{each task k}
						{
							Set $ y^j_{i,k}=r_{i_{\ell},k}^j+\gamma Q^{i,k}_{\phi_{i,k}^{\prime}}\left({s}^{{\prime}_j}_i, {a}^{{\prime}_j}_i\right) $\\
							Update local critics by minimizing the loss: \\$ \mathcal{L}(\phi_{i,k})=\frac{1}{S}\sum_{j}\left\{\left(Q^{i,k}_{\phi_{i,k}}\left({s}_i^j, {a}_i^j\right)-y^j_{i,k}\right)^{2}\right\} $\\
						}
						Update local actors: \\
						$
						\begin{aligned}
						\nabla J_{\theta_{i}}\approx
						&\frac{1}{S}\sum_{j}\Bigg\{\nabla_{\theta_{i}} \pi_{i}\left(a_{i} \mid s_{i}^j\right) \nabla_{a_{i}} Q_{\psi_{1}}^{g_{1}}\left(\mathbf{s}^j, \mathbf{a}^j \right) + \\
						& {\sum_{k=1}^{M}\left[\left.\nabla_{\theta_{i}} \pi_{i}\left(a_{i} \mid s_i^{j}\right) \nabla_{a_{i}} Q^{i,k}_{\phi_{i,k}}\left({s}^j_i, {a}^j_i \right)\right.\right]}\Bigg\}
						\end{aligned}
						$

						Update target networks parameters: \\
						\For{each task k}
						{
							$ \phi_{i,k}^{\prime} \leftarrow \tau\phi_{i,k} + (1-\tau)\phi_{i,k}^{\prime} $
						}
						$ \theta_{i}^{\prime} \leftarrow \tau\theta_{i} + (1-\tau)\theta_{i}^{\prime} $	
				}}
			}
		}
	\end{algorithm}
	and similarly, the agents' local critics are updated by (\ref{local_q}) and (\ref{local_a}). 
	The modified MADDPG framework depicted in Fig. \ref{MARL} is described in Algorithm \ref{algorithm1}. The core idea in TD3 is to delay the policy updates for $ d $ iterations until the convergence of value estimates. 
	Now, we can return to the issue of designing the reward function. The Reward function must judiciously be adjusted so that the multi-agent system steps on the path of solving the optimization problem (\ref{optimization_problem}). In essence, each PL as an agent, tries to access the available subchannels for two reasons: i) maintain an updated connection with the RSU and  keep the AoI level at its minimum, ii) disseminate the CAM information $ \zeta $ to its followers. Accordingly, we design the local reward of every platoon $ j $ as
	\begin{equation}\label{local_reward}
	\begin{aligned}
	r_{\ell}^j = &-\underbrace{\left\{\kappa_1{\zeta_j^r}/{\zeta_j}\right\}}_{\textit{Mode 1}} - \underbrace{\kappa_2A_j^t + \kappa_3G\left(\mathcal{C}_{j,\Re}^{t} - \mathcal{C}_{j,\Re}^{\text{min}}\right)}_{\textit{Mode 0}}\\
	&-\kappa_4\mathcal{F}\{p^t_j\},
	\end{aligned}
	\end{equation}
	where $ \kappa_1 - \kappa_4 $ are weighting factors used for balancing the reward, and $ \mathcal{F}\{.\} $  is a function that restricts the power quantity to the same range as the other components in the reward function.
	Furthermore, $ G(x) $ is a stepwise function given by
	\begin{equation}
	G(x)=\left\{\begin{array}{l}
	A, \quad x \geq 0, \\
	0, \quad x<0,
	\end{array}\right.
	\nonumber
	\end{equation}
	where $ A>0 $ is tuned to be a positive constant to indicate the revenue.
	The reward function in (\ref{local_reward}) consists of three parts that are matched with the objective function of the optimization problem (\ref{optimization_problem}): the first part is related to the reward the agent receives when the intra-platoon communication is chosen, the second part refers to the reward for the agent in the inter-platoon communication mode and the third part is related to the negative reward for the agent due to the power consumption. 
	Correspondingly, we define the global reward function as
	\begin{equation}\label{global_reward}
	r_g^t = - \frac{1}{P}\sum_{j\in\mathcal{P}}\sum_{k\in\mathcal{K}}\log_{10}\{{\mathbf{I}}^{t}_j[k]\}.
	\end{equation}
	
	The inspiration behind choosing the global reward function to be equal to the average interference is that the platoons are derived toward choosing subchannels and power levels that impose less interference on other platoons. 
	It is observed from Algorithm \ref{algorithm1} that the global critic is trained more than the local actor and critic networks since we have applied the TD3 algorithm. 
	The introduced delay, which is related to the hyperparameter $d$, can lead to faster convergence of the system by addressing the overestimation bias of global Q-function. 
	
	The following section will discuss the multi-task MARL, its corresponding formulations, and the intuition behind devising such an algorithm.
	
	\subsection{Modified MADDPG With Task Decomposition}
	In practice, the RL agents have to perform multiple tasks., and in order to drive the policy toward maximizing these tasks simultaneously, we have to integrate these tasks as a single holistic task and design a single reward signal, as was stated in (\ref{local_reward}). However, the drawback of applying such a method is that it cannot guarantee each sub-objective optimality, although the holistic reward function may exhibit encouraging signs of convergence. Therefore, for a MARL system consisting of $ M $  tasks and $ P $ agents, we change the optimization problem (\ref{marlopt}) as follows:
	\vspace*{-0.5em}
	\begin{equation}\label{marlopttd}
	\begin{aligned}
	& \max_{\pi_{j}}\mathcal{J}_j(\pi_j),\quad j\in\mathcal{P},~~ \pi_j\in\Pi_j\\
	& \mathcal{J}_j(\pi_j) = [\mathcal{J}^1_j(\pi_j),\dots,\mathcal{J}_j^M(\pi_j)],
	\end{aligned}	
	\end{equation}
	where $ \mathcal{J}_j^M(\pi_j) $ is related to the agent $ j $'s objective function for the $ M $th task. From (\ref{marlopttd}), it is conceived that we can deconstruct the holistic objective function into multiple sub-objectives based on the sub-tasks. The following theorem provides the condition for task decomposition, which results from decomposing the holistic reward function into sub-reward functions that can optimize the corresponding sub-objectives separately.

	\begin{thm}\label{theory1}
		If the reward function  $ R $, can be decomposed into $ M $ sub-reward functions, i.e., $ R(s, a, s^{\prime})=\sum_{k=1}^{M}r_k(s, a, s^{\prime}) $, then the holistic objective function $ \mathcal{J}_j(\pi_j) $ can be written as $ \mathcal{J}_j(\pi_j) =\sum_{k=1}^{M}\mathcal{J}_j^k(\pi_j) $, where 
		\begin{equation}
			\mathcal{J}_j^k(\pi_j) = \mathbb{E}\left[\sum_{t=0}^{\infty}\gamma^tr^{t+1,k}_{j}|s^0_j\right],\quad k=1,\dots,M.
		\end{equation}
	\end{thm}
	\begin{proof}
		 We refer the readers to Section III of \cite{van2017hybrid} for a extensive review of reward decomposition literature in RL.
	\end{proof}

	Based on Theorem \ref{theory1}, we can decompose the agents' local critics in (\ref{td3_dec_maddpg}) based on the sub-tasks,  and the resulting policy gradient considering the functionality of the global critic would be,
	
	\begin{equation}\label{td3_dec_maddpg_decomposition}
	\begin{aligned}
	\nabla_{\theta_j}\mathcal{J}_j=&\underbrace{\mathbb{E}_{\mathbf{s}, \mathbf{a}\sim\mathcal{D}}\left[\left.\nabla_{\theta_{j}} \pi_{j}\left(a_{j} \mid s_{j}\right) \nabla_{a_{j}} Q^{g_{1}}_{\psi_{1}}\left(\mathbf{s}, \mathbf{a} \right)\right.\right]}_{\textit{TD3 Global Critic}} + \\
	& \underbrace{\sum_{k=1}^{M}\mathbb{E}_{{s}_j, {a}_j\sim\mathcal{D}}\left[\left.\nabla_{\theta_{j}} \pi_{j}\left(a_{j} \mid s_{j}\right) \nabla_{a_{j}} Q^{j,k}_{\phi_{j,k}}\left({s}_j, {a}_j \right)\right.\right]}_{\textit{Decomposed Local Critics}},
	\end{aligned}
	\end{equation}
	where the parameters of sub-critics for agent $ j $ are updated as
	\begin{equation}
	\begin{aligned}
		&\mathcal{L}^j_k(\phi_{j,k})=\mathbb{E}_{\mathbf{s}_j, \mathbf{a}_j, \mathbf{r}_j, \mathbf{s}^{\prime}_j}\left[\left(Q^{j,k}_{\phi_{j,k}}\left({s}_j, {a}_j\right)-y^{j,k}_{\ell}\right)^{2}\right],\\
		&y^{j,k}_{\ell}=r^{j,k}_{\ell}+\left.\gamma Q^{j,k}_{\phi_{j,k}^{\prime}}\left({s}_j^{\prime}, {a}_j^{\prime}\right)\right|_{a_{j}^{\prime}=\pi_{j}^{\prime}\left(s_{j}^{\prime}\right)}, k=1,\dots,M.
	\end{aligned}
	\end{equation}
	Comparing (\ref{td3_dec_maddpg_decomposition}) with (\ref{td3_dec_maddpg}) reveals that
	\begin{equation}
		Q^{j}\left({s}_j, {a}_j \right)=\sum_{k=1}^{M}Q^{j,k}\left({s}_j, {a}_j \right),
	\end{equation}
	which can be easily derived from Theorem \ref{theory1}. In other words, the decomposition of the holistic reward function leads to the decomposition of the value functions. From the aforementioned analysis, the holistic reward function in (\ref{local_reward}) can be decomposed into two sub-reward functions as follows:
	\begin{itemize}
		\item {\textit{Task.~1 reward (CAM message transmission)}}
			\begin{equation}\label{task1reward}
					r_{\ell}^{j,1} = -\left\{\kappa_1{\zeta_j^r}/{\zeta_j}\right\}
			     	-\theta^t_j\kappa^{\prime}_4\mathcal{F}\{p^t_j\}.
			\end{equation}
		\item {\textit{Task.~2 reward (AoI minimization)}}
			\begin{equation}\label{task2reward}
				\begin{aligned}
					r_{\ell}^{j,2} = & -\kappa_2A_j^t + \kappa_3G\left(\mathcal{C}_{j,\Re}^{t} - \mathcal{C}_{j,\Re}^{\text{min}}\right)\\
					&-(1-\theta^t_j)\kappa^{\prime}_4\mathcal{F}\{p^t_j\},
				\end{aligned}
			\end{equation}
	\end{itemize}
	where $ \kappa^{\prime}_4 = \kappa_4 $ in (\ref{local_reward}). In other words we have $$ r_{\ell}^{j}=r_{\ell}^{j,1}+r_{\ell}^{j,2},\quad \forall j\in\mathcal{P}. $$ The general rule of thumb which governs the task decomposition procedure in (\ref{task1reward}) and (\ref{task2reward}) is that there should not be any temporal relationship between the sub-tasks. Due to this reason, we have considered the impact of power control in both the sub-reward functions as it influences both sides. The corresponding algorithm of modified MADDPG with task-decomposition is shown in Algorithm \ref{algorithm2}. In the next section, we will investigate the complexities of the proposed algorithms and assess the growth of the parameter space as the number of agents increases.
	\subsection{Computational Complexity}
	The computational complexity is crucial to the utility of the algorithms. Therefore, we analyze the computational complexity of the two proposed RL methods and compare them with the conventional MADDPG framework, which is extensively applied in the literature. In essence, these analytics depend on two parameters, i) the number of trainable parameters, ii) the total number of neural networks used in the algorithms.
	
	\textit{i) The number of trainable parameters:}\\
	In MADDPG, the centralized Q-functions take all the agents' observations and actions as their input. Concretely, assuming all the agents have identical observation and action
	spaces shown by $ \omega $ and $ \alpha $, the number of trainable parameters for the MADDPG method would be $ \mathcal{O}(n^2(\omega + \alpha)) $, where $ n $ indicates the number of agents. Conversely, the two proposed RL methods incorporate two types of critic networks: the global and local critics. Both the algorithms share a global centralized Q-function whose parametric space increases linearly and is represented as  $ \mathcal{O}(n(\omega + \alpha)) $. On the other hand, the local critics in the two RL methods only take the respective agent's observation and action as their input. Consequently, the parameter space of local critics can be expressed as $ \mathcal{O}(\omega + \alpha) $, and this is similar for both the algorithms.
	
	\textit{ii) The total number of neural networks}: \\
	In MADDPG, the total number of neural networks used during the training process is equal to $ 2\times(n(\underline{1}_{\mathtt{Q}}+\underline{1}_{\mathtt{A}})) $, where the multiplication by 2 is because of the target networks, $ \underline{1}_{\mathtt{Q}} $, and $ \underline{1}_{\mathtt{A}} $ represents that there is one critic and actor network specific for each agent, and $ n $ is the total number of agents. For the modified MADDPG framework, the total number of neural networks is $ 2\times(n(\underline{1}_{\mathtt{Q}_{\ell}}+\underline{1}_{ \mathtt{A}_{\ell} })+\underline{1}_{ \mathtt{Q}_{g} }) $, where $ \underline{1}_{ \mathtt{Q}_{g} } $ indicates the total number of global critics. It is worth mentioning that applying the TD3 algorithm doubles the number of global critics, and in this case the number of neural networks will be $ 2\times(n(\underline{1}_{\mathtt{Q}_{\ell}}+\underline{1}_{ \mathtt{A}_{\ell} })+\underline{2}_{ \mathtt{Q}_{g} }) $. Finally, for the modified MADDPG method with task decomposition, number of neural networks will be 
	$ 2\times(n(\underline{k}_{\mathtt{Q}_{\ell}}+\underline{1}_{ \mathtt{A}_{\ell} })+\underline{1}_{ \mathtt{Q}_{g} }) $, where $ \underline{k}_{\mathtt{Q}_{\ell}} $ indicates that there is a separate Q-function for each agent's decomposed tasks. Similarly, this number will be $ 2\times(n(\underline{k}_{\mathtt{Q}_{\ell}}+\underline{1}_{ \mathtt{A}_{\ell} })+\underline{2}_{ \mathtt{Q}_{g} }) $, whenever the TD3 algorithm is further applied.
		\begin{table}[!t]
		\vspace*{-\baselineskip}
		\caption{Simulation Parameters}
		\label{Table_parameters_s}
		\resizebox{1\columnwidth}{!}{
			\begin{tabular}{p{5.7cm} p{2.7cm}}
				\toprule[1pt]
				\textbf{Vehicular environment parameters} & \textbf{Value}  \\
				\midrule
				Carrier frequency& 2 GHz\\
				Number of RBs & 3\\
				Bandwidth of each RB& 180 kHz\\
				Number of Vehicles &  16 -- 50 \\
				Size of Platoons   &  4 -- 10\\
				Platoons Speed	   &  36 -- 54 km/h\\
				Intra-platoon gap&     5, 15, 25, 35 m\\
				RSU and vehicles antenna heights& 25, 1.5 m\\
				RSU and vehicles antenna gains& 8, 3 dBi\\
				RSU and vehicles receiver noise figure& 5, 9 dB\\
				Vehicles mobility model& Urban case of A.1.2 \cite{3gpp_36_885}\\
				Vehicles maximum power& 30 dBm\\
				Noise power $ \sigma^2 $& -114 dBm\\
				Time constraint of CAM dissemination, $ T $,& 100 ms\\
				CAM message size& 4000 bytes\\
				V2I links\footnotemark[1] minimum capacity requirement, $ \mathcal{C}_{j,\Re}^{\text{min}} $ & 3 bps/Hz\cite{3gpp_37_885}\\
				V2I links path loss model& $ 128.1+37.6\log_{10}(d) $\\
				V2V links\footnotemark[2] path loss model& LOS in WINNER+ B1 Manhattan \cite{bultitude20074}\\
				Shadowing distribution& Log-normal\\
				Shadowing standard deviation for V2I links& 8 dB\\
				Shadowing standard deviation for V2V links& 3 dB\\
				Decorrelation distance for V2I/V2V links&50, 10 m\\ 
				Pathloss/shadowing update for V2I/V2V links&Every 100 ms \cite{3gpp_36_885}\\ 
				Fast fading update for V2I/V2V links&Every 1 ms \cite{3gpp_36_885}\\
				Fast fading& Rayleigh fading\\
				\toprule[1pt]
				\textbf{Neural networks parameters} & \textbf{Value}  \\
				\midrule
				Experience replay buffer size & 50000 \\
				Mini batch size&64\\
				Number/size of local actor networks hidden layers&2 / 1024, 512\\
				Number/size of local critic networks hidden layers&2 / 512, 256\\
				Number/size of global critic hidden layers&3/ 1024, 512, 256\\
				Critic/Actor networks learning rate&0.001/0.0001\\
				Discount factor&0.99\\
				Target networks soft update parameter, $ \tau $&0.0005\\
				Number of episodes&500\\
				Number of iterations per episode&100\\
				\bottomrule[1pt]
		\end{tabular}}
		\footnotemark[1]{Link between PL and RSU}\quad\quad
		\footnotemark[2]{Link between PL and its followers}
	\end{table}

	\section{Performance Evaluation}\label{Result}
	In this section, we assess the simulation results to validate the proposed multi-agent RL based resource allocation for the platooning system. We have built our simulation following the urban case defined in Annex A of \cite{3gpp_36_885}. Major simulation parameters, including the channel models for V2I and V2V links, are listed in Table \ref{Table_parameters_s}. 
	In addition, the Gaussian noise $ \epsilon \sim \mathcal{N}(0, 0.2) $ is added to the actions chosen by the target actor networks, and then clipped to $ (-0.5, 0.5) $ to smooth the target policy, and the policy update delay factor is set to $ d=2 $. Throughout the simulations, the number of available RBs is fixed to three; however, we have varied the number of platoons, the number of PMs, and the intra-platoon spacing to investigate their impact on the system's overall performance. It is worthwhile to mention that we fix the large-scale fading during each episode and let the small-scale fading alter; therefore, the RL algorithm can better procure the underlying fading dynamics. Due to the sensitivity of RL algorithms to the reward quantity, the global reward function in (\ref{global_reward}) is normalized to be consistent with the local reward's range.
	Furthermore, to verify our proposed method's efficiency, three algorithms are adopted as baselines:
	\begin{itemize}
		\item \textbf{Baseline 1: Modified MADDPG}
		
		This algorithm was introduced earlier in Algorithm \ref{algorithm1} and is shown to outperform the conventional MADDPG in \cite{sheikh2020multi}. In this algorithm, the global critic implemented in the RSU motivates cooperation between the platoons by periodically reporting the effectiveness of platoons' chosen action.  The local critics and actor networks are implemented in each platoon and trained with each platoon's local training dataset without the need for other platoons' information.
		
		\item \textbf{Baseline 2: Fully decentralized MADDPG}
		
		To illustrate the global critic's impact on the network performance, in this algorithm, the global critic is not taken into account, and the platoons choose their actions in a fully decentralized way, based on their observations.
		
		\item \textbf{Baseline 3: DDPG}
		
		In this algorithm, the RSU has to acquire all the platoons' observations and actions and is considered a fully centralized algorithm in which all the computations and decision-making have to be performed in the RSU. 
	\end{itemize}
	
		\begin{figure*}[!t]
		\vspace*{-\baselineskip} 
		\centering
		\includegraphics[width=1\textwidth]{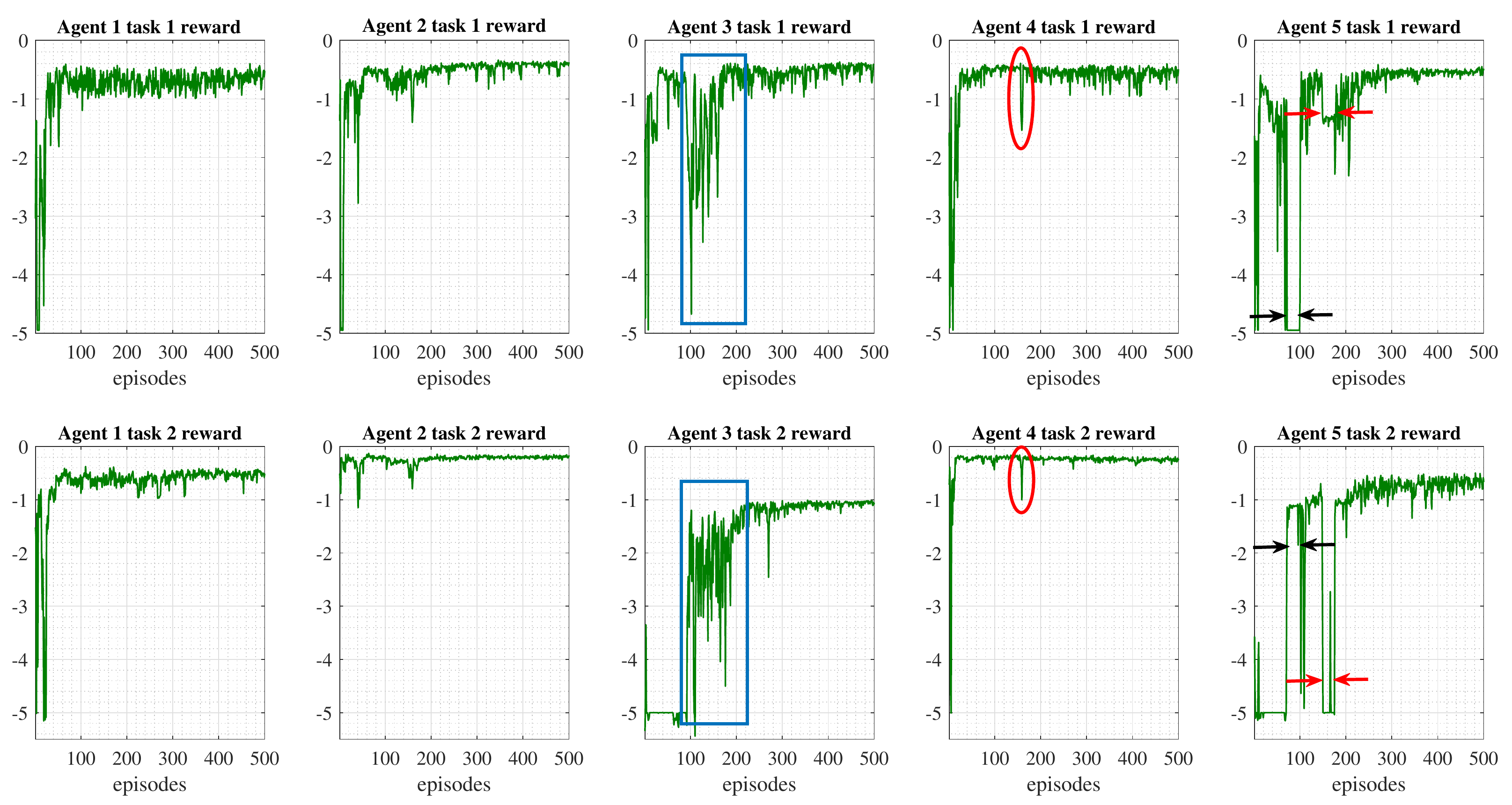}
		\caption{Convergence performance of agents sub-tasks following Algorithm \ref{algorithm2}, intra-platoon gap = 25m, platoon size = 6.} 
		\label{Re1}
	\end{figure*}
	
	\begin{figure*}[!t]
		\centering     
		\subfigure[$ P=5 $, $ N=4, $ intra-platoon gap $ = 5 m $ ]{\label{Re2}\includegraphics[width=.49\textwidth]{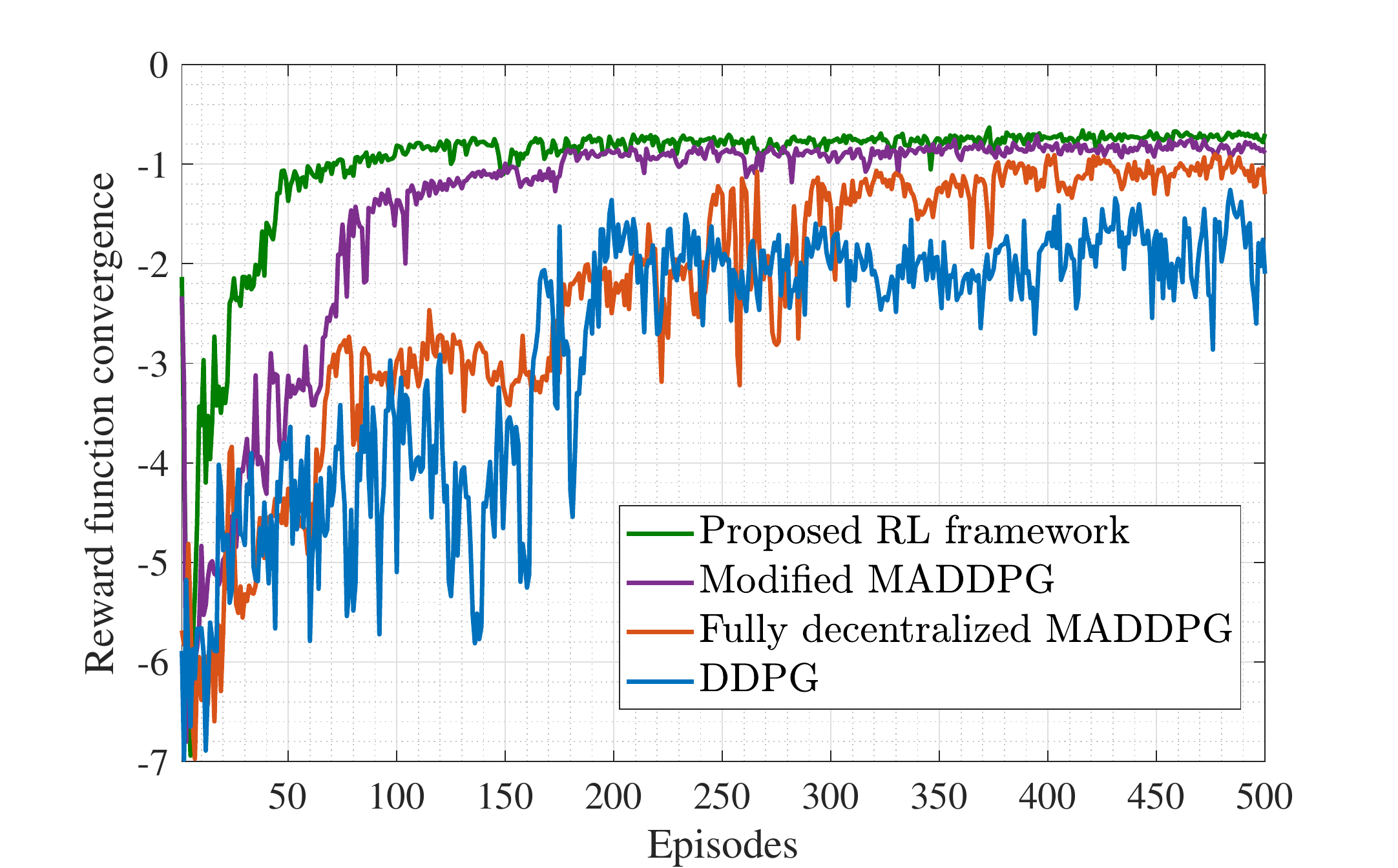}}
		\subfigure[$ P=7 $, $ N=4, $ intra-platoon gap $ = 5 m $]{\label{Re3}\includegraphics[width=.49\textwidth]{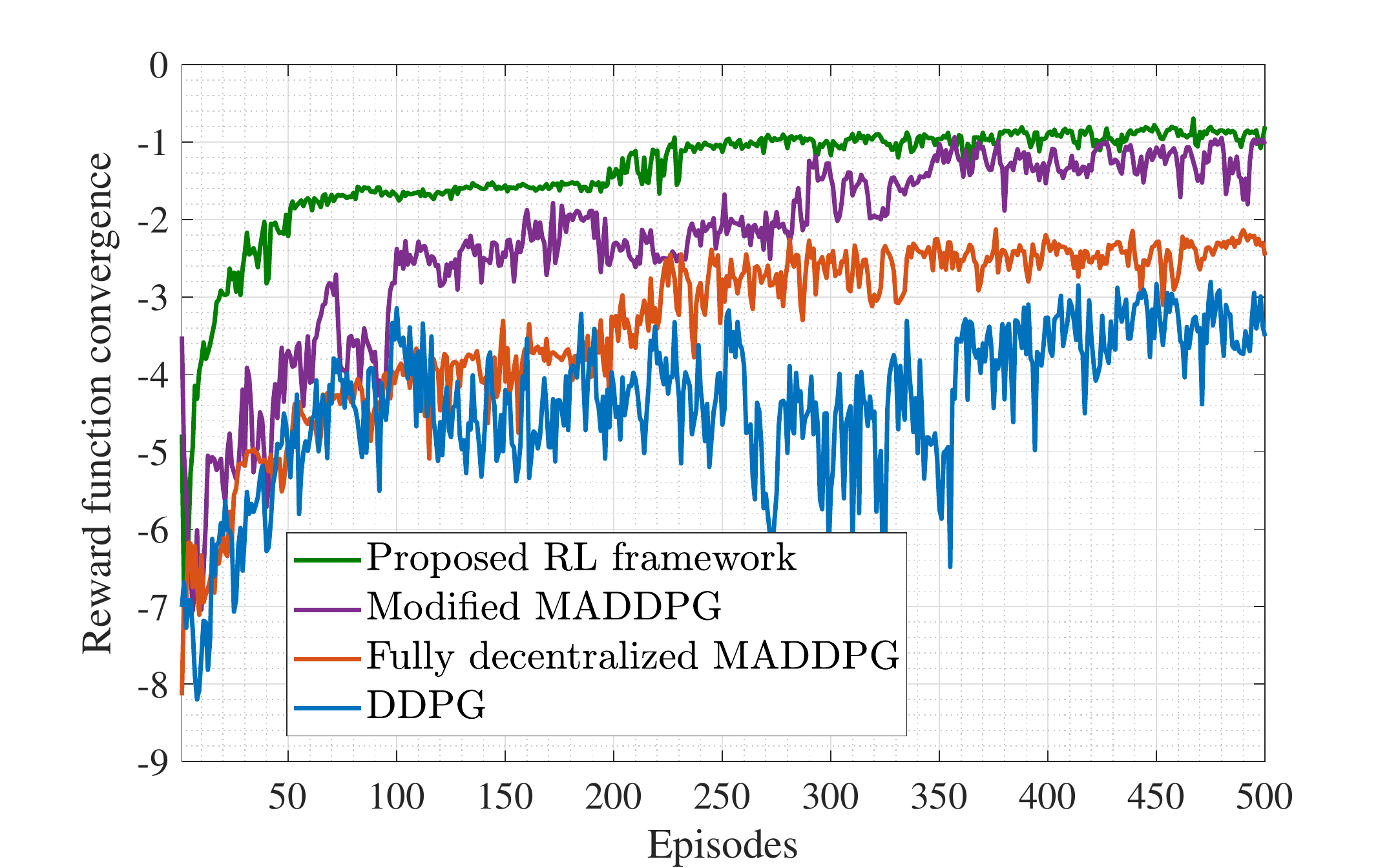}}
		\caption{Comparison of convergence performance}
		\label{REEE}
	\end{figure*}
	\subsection{Simulation Results}
	Fig. \ref{Re1} indicates the convergence of agents sub-tasks when the intra-platoon gap is 25 m, and the number of platoon members at each platoon is 6 (30 vehicles in total). For each agent, we have plotted its sub-tasks reward function. Two notable trends stand out in the figure; first, it can be seen that all the agents have been able to fulfill their associated tasks and maximize the designated reward functions in (\ref{task1reward}) and (\ref{task2reward}) during the $ T $ seconds. Second, the proposed algorithm is quite fast in convergence time. It is observed that for most of the agents, the task-wise reward functions converge in less than 50 episodes. In addition to some fluctuations due to the channel fading that arose by platoons' movements in the environment, the following observations can also be noticed. Since the number of vehicles is large compared to the available resources, there is high contention between the platoons in terms of accessing the available resources. Therefore, the platoons have to share the resources. However, they have to control their power usage jointly with the mode they choose to operate so as not to impose much interference to the other platoons reusing the same resources. This issue is of paramount importance as the platoons choose their actions based on their own observations. The figure implicitly indicates that different components of the system have somehow reached an equilibrium. In other words, not only the global critic has been able to drive the platoons toward selecting proper resources to impose less interference on each other, but also the local critics have motivated their respective platoons to flexibly alter their decisions between inter and intra-platoon modes and meet the predetermined requirements. 	
	
	Fig. \ref{Re1} also reveals that the number of episodes agent three and agent five
	needed for proper convergence is longer compared to the other agents. Starting with agent three, it is observed that during the first 100 episodes, agent three has focused only on task one (CAM dissemination, $\theta^t_j=1$), which has led to a destructive reward for task two. This irregular functionality, which has stemmed from the destructive actions chosen by the agent's actor network, is feedbacked through  (\ref{td3_dec_maddpg_decomposition}) into the agent's actor network to update the policy toward better performance. As the policy starts to improve, agent three, like the other agents, begins to exhibit encouraging signs as it is highlighted in a blue rectangle for both of its tasks. The same procedure applies for agent five. During the first 200 episodes, agent five has focused only on one of its tasks leading to an increase in one task's reward and a substantial decrease in the other one. These fluctuations are demonstrated with red and black arrows for agent five's task-wise reward functions. Aside from what is explained so far, for some episodes, there can be seen some unusual bounces in the reward function, one of which is marked with a red ellipsoid in agent four's reward function, which can be partly related to a phenomenon called catastrophic forgetting in neural networks \cite{mccloskey1989catastrophic}. In general, the proposed MARL method has robust functionality, and yields compelling results even in complex environments consisting of even more vehicles.	
		\begin{figure*}[!t]
		\centering     
		\subfigure[Average Age of Information versus the intra-platoon gap for  P=5 ,  N=4 ]{\label{Re4}\includegraphics[width=.49\textwidth]{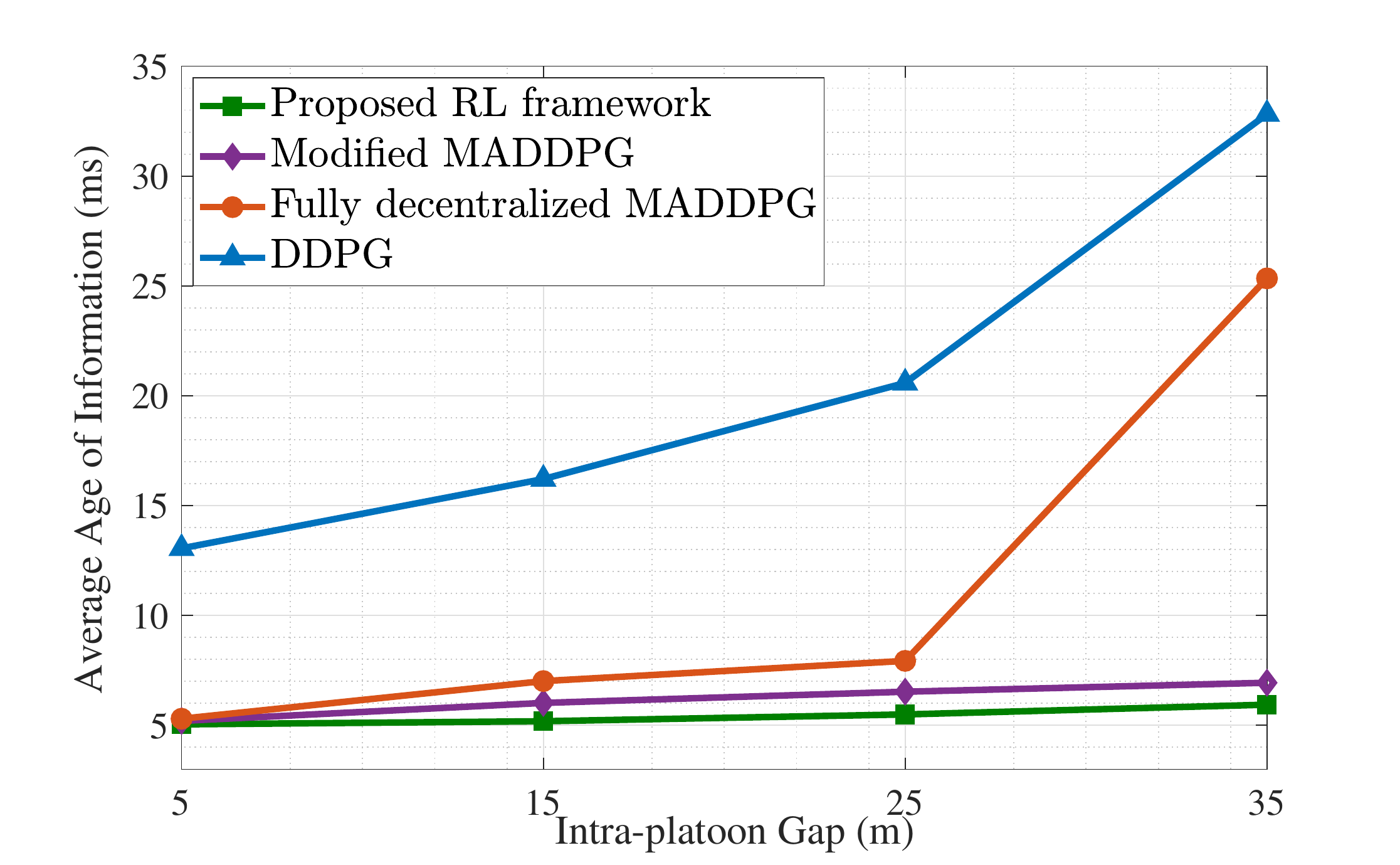}}
		\subfigure[Average CAM message transmission probability versus the intra-platoon gap for  P=5 ,  N=4 ]{\label{Re5}\includegraphics[width=.49\textwidth]{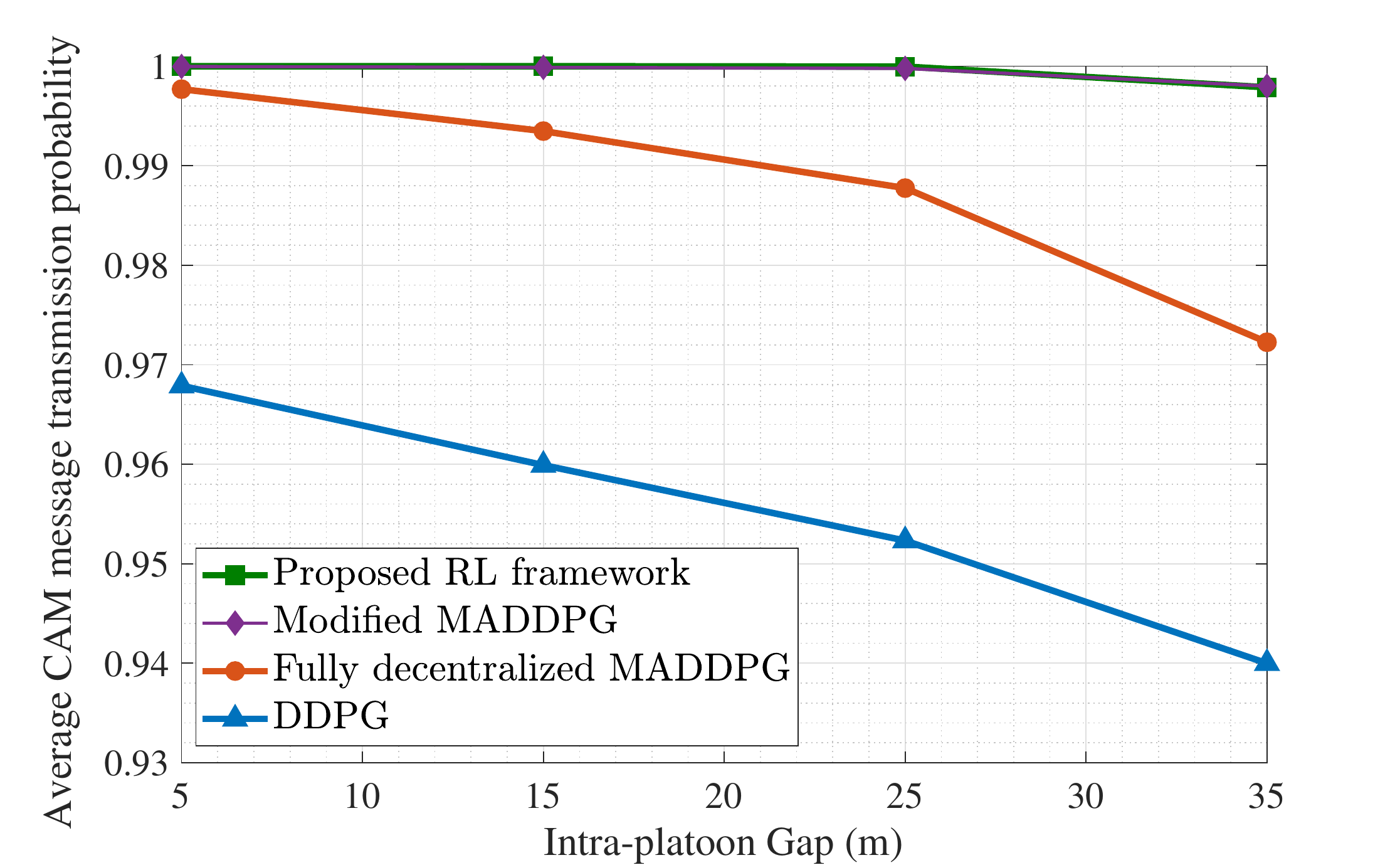}}
		\subfigure[Average Age of Information versus the the number of platoon members for  P=5 , intra-platoon gap =  25 m ]{\label{Re6}\includegraphics[width=.49\textwidth]{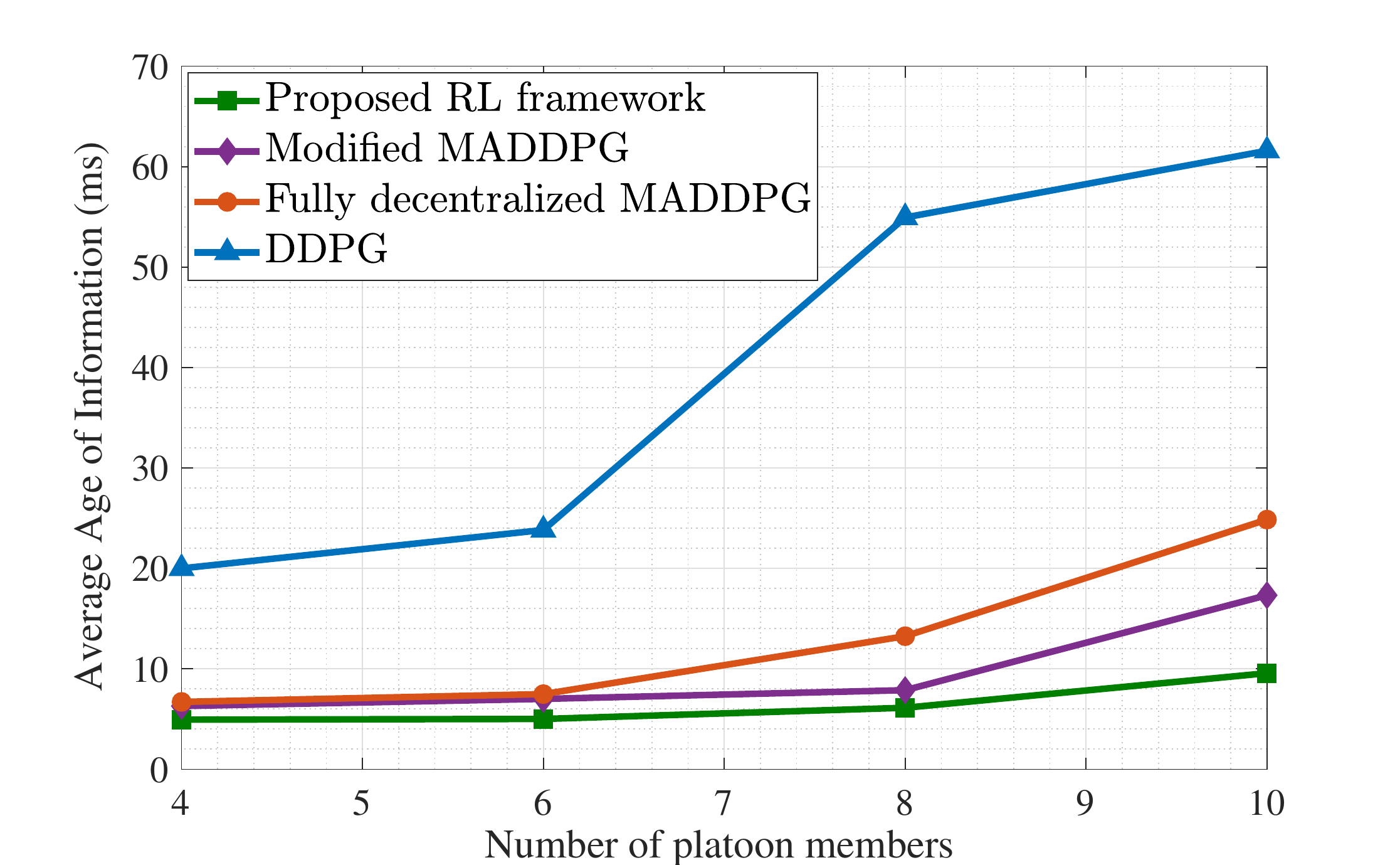}}
		\subfigure[Average CAM message transmission probability versus the number of platoon members for  P=5 , intra-platoon gap =  25 m ]{\label{Re7}\includegraphics[width=.49\textwidth]{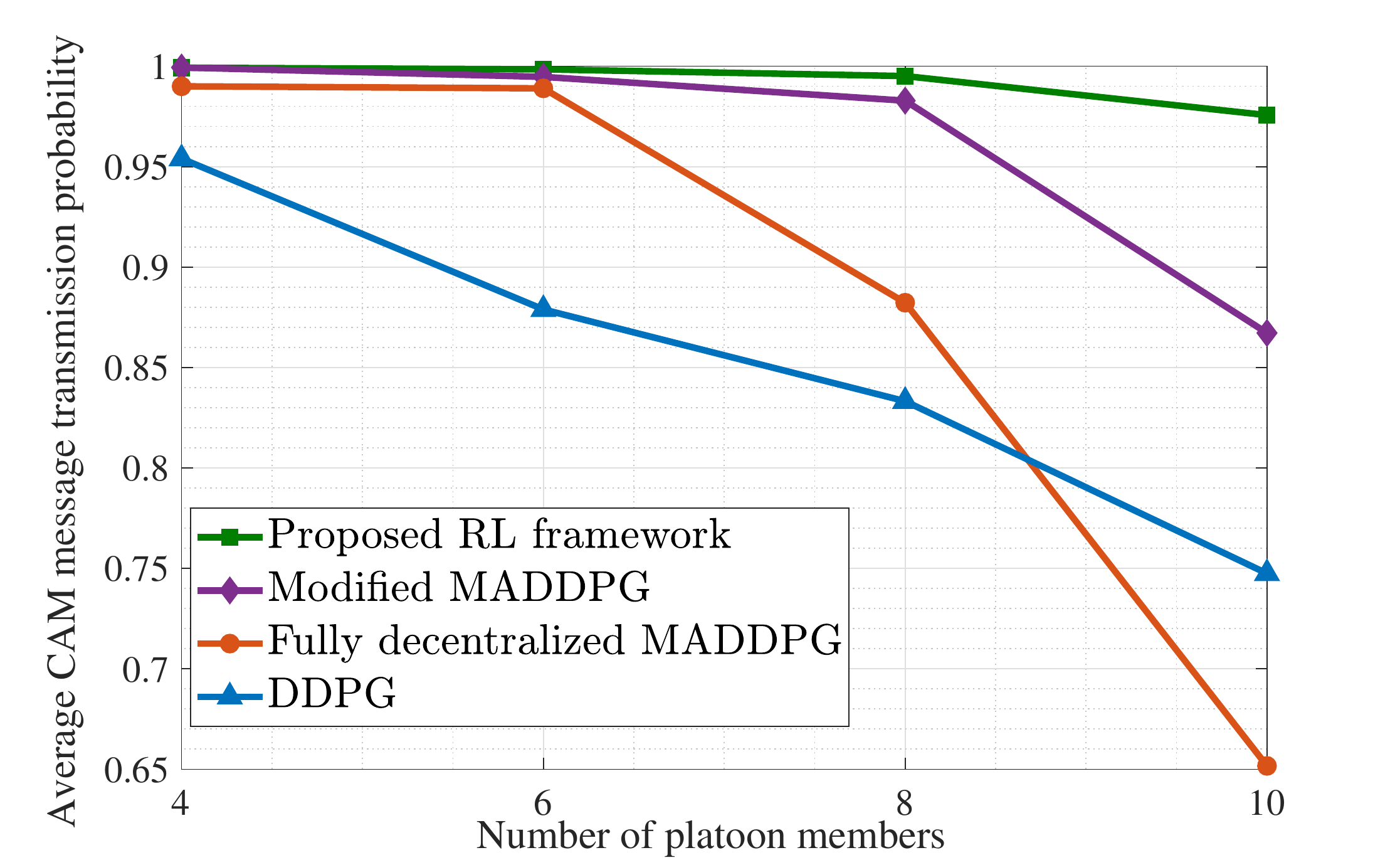}}
		\caption{Comparison of proposed RL algorithms in terms of Average Age of Information and CAM message transmission probability}
		\label{REEE1}
	\end{figure*}
	
	Fig. \ref{REEE} compares the convergence of the four approaches in terms of the average reward performance when the number of platoons is five and seven, respectively. At first glance, the proposed method indisputably outperforms the other three baselines. The DDPG method has the worst reward performance among the considered RL algorithms in both figures. The reason for this weak execution can be related to the DDPG's centralized behavior. Since the DDPG has to take all the agents' observations and actions as input and evaluate how decent the policy has performed for all the agents, it fails to address the agents' individual performance and acts non-stationarily in multi-agent environments. This improper execution is further intensified with the number of agents in Fig. \ref{Re3}.
	
	Regarding the fully decentralized MADDPG, the agents act absolutely oblivious without any knowledge about other existing agents' policies or actions in the environment. This unawareness can lead to increased levels of interference in the system, which will degrade the agents' overall performance. This phenomenon is not very severe when the number of platoons is low, as can be seen from Fig. \ref{Re2}; however, by increasing the number of platoons, its tendency even to perform worse than DDPG is not inconceivable, as observed from Fig. \ref{Re3}. 
	One prominent feature that separates the modified MADDPG and proposed RL framework from the fully decentralized and DDPG, aside from their better reward performance and faster convergence, is their stability and minimal fluctuations during the convergence. We can summarize the primary reasons for this performance gap as follows: 
	i) The proposed frameworks can learn to maximize the individual and global rewards for all the agents simultaneously, leading to improved collaboration between the agents, hence driving towards better performance. ii) The global critic, which is based on TD3, considers the correspondence between function approximation error in both policy and value updates. On the other hand, the DDPG method is highly susceptible to inaccuracies provoked by function approximation errors, making it overfit to narrow peaks in the value estimate. iii) Last but not least, unlike the original implementation of DDPG, which leverages the correlated Ornstein-Uhlenbeck noise, the proposed MARL framework applies an uncorrelated Gaussian noise for exploration.
	Eventually, by analyzing Figs. \ref{Re2} and \ref{Re3}, it is unveiled that the proposed RL algorithm tends to converge to the same quantity even though the vehicle density has increased in the environment, while the other baselines' performance diminishes with the increased load.
	
	Fig. \ref{Re4} illustrates the mean AoI of platoons as a function of the intra-platoon gap when P = 5 and N = 4. From the figure, it can be observed that the AoI quantity rises for all the considered algorithms as the intra-platoon spacing increases. The intuition behind this observation is straightforward.  By increasing the intra-platoon gap, it is perceptible that the channel conditions from PLs to their followers sustain more variations, leading to lower data rates. Accordingly, the PL spends more time transmitting the CAM message to its followers and operating in Mode 1. In the meantime, the PL less frequently interacts with the RSU; therefore, the average AoI increases.
	Nonetheless, our proposed MARL framework performs significantly more reliable than the other baselines, maintaining the average AoI quantity within 5-10 milliseconds range, and guarantees better QoS. Stunningly, the modified MADDPG framework acts close to our proposed algorithm. This behavior is anticipated as both the algorithms leverage the global and local critics simultaneously to learn a global and individual reward. However, there is still a slight performance gap between them due to the task decomposition in our proposed algorithm.
	In comparison, the DDPG acts less stable, and its performance degrades by increasing the intra-platoon spacing. It is also observed that the performance of fully decentralized MADDPG is close to the modified MADDPG, and our proposed algorithm up to 25 meters intra-platoon gap; however, there can be seen a sharp jump in the AoI quantity when the intra-platoon gap rises to 35 meters. This is because, with longer distances between the PL and its followers, the PLs tend to use more power to compensate for the reduced levels of channel gains to guarantee the CAM message transmission to their followers, which inevitably results in severe interference to other platoons, and as these platoons are acting in a fully decentralized way, they cannot discern the appropriate resources to select, hence leading to these sharp changes in the performance metrics.
	The aforementioned analysis is also extendible to results in Fig. \ref{Re6}, which demonstrates the average AoI versus the number of platoon followers. 
	
	Another compelling result can be observed from Figs. \ref{Re5} and \ref{Re7}, which show the CAM message transmission probability. From the figures, the performance metric drops for all the schemes as the intra-platoon gap increases. In conjunction with the observations from Figs. \ref{Re4} and \ref{Re6}, the intuition behind this phenomenon is explicit. However, as Figs. \ref{Re5} and \ref{Re7} suggest the proposed framework is robust against alterations in platoon sizes or intra-platoon spacing variations. The proposed framework maintains a transmission probability of over 99 percent for different platoon sizes when the intra-platoon gap is less than 25 meters, whereas this metric drops significantly for DDPG and fully decentralized MADDPG. We finalize the respective analysis with a critical look at all four figures. By comparing Figs. \ref{Re4} - \ref{Re7}, it is conceivable that the number of vehicles significantly impacts the performance metrics quantity. In Fig. \ref{Re7}, by increasing the number of vehicles up to 30, except DDPG, all the algorithms have shown a similar behavior. As we continue increasing the number of vehicles up to 50, the gap between these algorithms starts to grow. One interesting observation from this figure is that the CAM message transmission probability has dropped to 65 percent for fully decentralized MADPG, even worse than DDPG, which directly relates to its lack of interference management when the number of vehicles is considerable. Nevertheless, care must be taken since these observations are based on the particular setting for the simulation, and additional caution is required when generalizing them.  We can still conclude that our proposed framework in Algorithm \ref{algorithm2} indicated a very robust behavior against the parameter modifications and outperformed the other baselines.
	
	\section{Conclusion}\label{Conclusion}
	In this paper, a MADDPG-based transmission mode selection and resource allocation method was developed for platooning systems, aiming at minimizing the AoI of platoons while guaranteeing the CAM message delivery to PMs. The proposed MARL framework consists of a collaborative setting where a group of PLs simultaneously learns to maximize the collective global reward and individual local reward. 
	Furthermore, we decomposed the agents' holistic reward signal into multiple sub-reward functions based on their sub-tasks and learned them separately.
	Through such a mechanism, we demonstrate that the proposed RRM scheme is significantly robust and effective in encouraging platoons to improve system-level performance, although the PLs independently select their transmission mode, RB, and power levels. Finally, through extensive simulations we verified the effectiveness and performance of the proposed MARL method.
	Future work will carry an in-depth extension of the proposed framework to Non-orthogonal Multiple Access (NOMA) scenarios for the platooning system. Also, examining the spectrum sharing scenarios in vehicular networks is another encouraging direction worth further investigation.
	
	\newpage
	\bibliographystyle{IEEEtran}
	\bibliography{References}

\begin{thebibliography}{10}
\providecommand{\url}[1]{#1}
\csname url@samestyle\endcsname
\providecommand{\newblock}{\relax}
\providecommand{\bibinfo}[2]{#2}
\providecommand{\BIBentrySTDinterwordspacing}{\spaceskip=0pt\relax}
\providecommand{\BIBentryALTinterwordstretchfactor}{4}
\providecommand{\BIBentryALTinterwordspacing}{\spaceskip=\fontdimen2\font plus
\BIBentryALTinterwordstretchfactor\fontdimen3\font minus
  \fontdimen4\font\relax}
\providecommand{\BIBforeignlanguage}[2]{{%
\expandafter\ifx\csname l@#1\endcsname\relax
\typeout{** WARNING: IEEEtran.bst: No hyphenation pattern has been}%
\typeout{** loaded for the language `#1'. Using the pattern for}%
\typeout{** the default language instead.}%
\else
\language=\csname l@#1\endcsname
\fi
#2}}
\providecommand{\BIBdecl}{\relax}
\BIBdecl

\bibitem{zhang2011data}
J.~Zhang, F.-Y. Wang, K.~Wang, W.-H. Lin, X.~Xu, and C.~Chen, ``{Data-driven
  intelligent transportation systems: A survey},'' \emph{IEEE Transactions on
  Intelligent Transportation Systems}, vol.~12, no.~4, pp. 1624--1639, December
  2011.

\bibitem{hall2005vehicle}
R.~Hall and C.~Chin, ``{Vehicle sorting for platoon formation: Impacts on
  highway entry and throughput},'' \emph{Transportation Research Part C:
  Emerging Technologies}, vol.~13, no. 5-6, pp. 405--420, October 2005.

\bibitem{lioris2017platoons}
J.~Lioris, R.~Pedarsani, F.~Y. Tascikaraoglu, and P.~Varaiya, ``Platoons of
  connected vehicles can double throughput in urban roads,''
  \emph{Transportation Research Part C: Emerging Technologies}, vol.~77, pp.
  292--305, April 2017.

\bibitem{jia2015survey}
D.~Jia, K.~Lu, J.~Wang, X.~Zhang, and X.~Shen, ``{A survey on platoon-based
  vehicular cyber-physical systems},'' \emph{IEEE communications surveys \&
  tutorials}, vol.~18, no.~1, pp. 263--284, March 2015.

\bibitem{ETSI_CAM_STANDARD}
``{Intelligent Transport Systems (ITS); Vehicular Communications; Basic Set of
  Applications; Part 2: Specification of Cooperative Awareness Basic
  Service},'' {European Telecommunications Standards Institute (ETSI)},
  European Standard (EN) 302 637-2, September 2014, version 1.3.2.

\bibitem{rios2016survey}
J.~Rios-Torres and A.~A. Malikopoulos, ``{A survey on the coordination of
  connected and automated vehicles at intersections and merging at highway
  on-ramps},'' \emph{IEEE Transactions on Intelligent Transportation Systems},
  vol.~18, no.~5, pp. 1066--1077, May 2016.

\bibitem{sensors_platoon_v2x}
G.~Nardini, A.~Virdis, C.~Campolo, A.~Molinaro, and G.~Stea, ``{Cellular-V2X
  communications for platooning: Design and evaluation},'' \emph{Sensors},
  vol.~18, no.~5, p. 1527, May 2018.

\bibitem{3gpp_36_885}
``{Study on LTE-based V2X Services},'' {3rd Generation Partnership Project
  (3GPP)}, Technical Specification (TS) 36.885, June 2016, version 14.0.0.

\bibitem{3gpp_23_303}
``{Proximity-based services (ProSe)},'' {3rd Generation Partnership Project
  (3GPP)}, Technical Specification (TS) 23.303, July 2020, version 16.0.0.

\bibitem{3gpp_36_785}
``{Vehicle to Vehicle (V2V) services based on LTE sidelink},'' {3rd Generation
  Partnership Project (3GPP)}, Technical Specification (TS) 36.785, October
  2016, version 14.0.0.

\bibitem{3gpp_22_886}
``{Study on enhancement of 3GPP support for 5G V2X services},'' {3rd Generation
  Partnership Project (3GPP)}, Technical Report (TR) 22.886, December 2018,
  version 16.2.0.

\bibitem{molina2017lte}
R.~Molina-Masegosa and J.~Gozalvez, ``{LTE-V for sidelink 5G V2X vehicular
  communications: A new 5G technology for short-range vehicle-to-everything
  communications},'' \emph{IEEE Vehicular Technology Magazine}, vol.~12, no.~4,
  pp. 30--39, December 2017.

\bibitem{LTE-V}
S.~Chen, J.~Hu, Y.~Shi, and L.~Zhao, ``{LTE-V: A TD-LTE-based V2X solution for
  future vehicular network},'' \emph{IEEE Internet of Things journal}, vol.~3,
  no.~6, pp. 997--1005, September 2016.

\bibitem{Better_Platooning}
C.~Campolo, A.~Molinaro, G.~Araniti, and A.~O. Berthet, ``{Better platooning
  control toward autonomous driving: An LTE device-to-device communications
  strategy that meets ultralow latency requirements},'' \emph{IEEE Vehicular
  Technology Magazine}, vol.~12, no.~1, pp. 30--38, January 2017.

\bibitem{mei2018joint}
J.~Mei, K.~Zheng, L.~Zhao, L.~Lei, and X.~Wang, ``{Joint radio resource
  allocation and control for vehicle platooning in {LTE-V2V} network},''
  \emph{IEEE Transactions on Vehicular Technology}, vol.~67, no.~12, pp.
  12\,218--12\,230, December 2018.

\bibitem{wang2019platoon}
P.~Wang, B.~Di, H.~Zhang, K.~Bian, and L.~Song, ``{Platoon cooperation in
  cellular {V2X} networks for {5G} and beyond},'' \emph{IEEE Transactions on
  Wireless Communications}, vol.~18, no.~8, pp. 3919--3932, August 2019.

\bibitem{zeng2019joint}
T.~Zeng, O.~Semiari, W.~Saad, and M.~Bennis, ``{Joint communication and control
  for wireless autonomous vehicular platoon systems},'' \emph{IEEE Transactions
  on Communications}, vol.~67, no.~11, pp. 7907--7922, November 2019.

\bibitem{peng2017resource}
H.~Peng, D.~Li, Q.~Ye, K.~Abboud, H.~Zhao, W.~Zhuang, and X.~Shen, ``{Resource
  allocation for cellular-based inter-vehicle communications in autonomous
  multiplatoons},'' \emph{IEEE Transactions on Vehicular Technology}, vol.~66,
  no.~12, pp. 11\,249--11\,263, December 2017.

\bibitem{wang2018resource}
R.~Wang, J.~Wu, and J.~Yan, ``{Resource allocation for D2D-enabled
  communications in vehicle platooning},'' \emph{IEEE Access}, vol.~6, pp.
  50\,526--50\,537, September 2018.

\bibitem{zia2019distributed}
K.~Zia, N.~Javed, M.~N. Sial, S.~Ahmed, A.~A. Pirzada, and F.~Pervez, ``{A
  distributed multi-agent RL-based autonomous spectrum allocation scheme in
  {D2D} enabled multi-tier HetNets},'' \emph{IEEE Access}, vol.~7, pp.
  6733--6745, January 2019.

\bibitem{yang2019intelligent}
H.~Yang, X.~Xie, and M.~Kadoch, ``{Intelligent resource management based on
  reinforcement learning for ultra-reliable and low-latency IoV communication
  networks},'' \emph{IEEE Transactions on Vehicular Technology}, vol.~68,
  no.~5, pp. 4157--4169, May 2019.

\bibitem{sun2019application}
Y.~Sun, M.~Peng, Y.~Zhou, Y.~Huang, and S.~Mao, ``{Application of machine
  learning in wireless networks: {Key} techniques and open issues},''
  \emph{IEEE Communications Surveys \& Tutorials}, vol.~21, no.~4, pp.
  3072--3108, June 2019.

\bibitem{V2X_Modeselection_DRL}
X.~{Zhang}, M.~{Peng}, S.~{Yan}, and Y.~{Sun}, ``{{Deep} reinforcement
  learning-based mode selection and resource allocation for Cellular {V2X}
  communications},'' \emph{IEEE Internet of Things Journal}, vol.~7, no.~7, pp.
  6380--6391, July 2020.

\bibitem{Multi_Agent_DRL_Urban}
T.~{Wu}, P.~{Zhou}, K.~{Liu}, Y.~{Yuan}, X.~{Wang}, H.~{Huang}, and D.~O. {Wu},
  ``{Multi-agent deep reinforcement learning for urban traffic light control in
  vehicular networks},'' \emph{IEEE Transactions on Vehicular Technology},
  vol.~69, no.~8, pp. 8243--8256, August 2020.

\bibitem{ye2019deep}
H.~Ye, G.~Y. Li, and B.-H.~F. Juang, ``Deep reinforcement learning based
  resource allocation for {V2V} communications,'' \emph{IEEE Transactions on
  Vehicular Technology}, vol.~68, no.~4, pp. 3163--3173, April 2019.

\bibitem{An_Intelligent_Path_Planning_DRL}
C.~{Chen}, J.~{Jiang}, N.~{Lv}, and S.~{Li}, ``An intelligent path planning
  scheme of autonomous vehicles platoon using deep reinforcement learning on
  network edge,'' \emph{IEEE Access}, vol.~8, pp. 99\,059--99\,069, May 2020.

\bibitem{vu2020multi}
H.~V. Vu, Z.~Liu, D.~H. Nguyen, R.~Morawski, and T.~Le-Ngoc, ``Multi-agent
  reinforcement learning for joint channel assignment and power allocation in
  platoon-based {C-V2X} systems,'' \emph{arXiv preprint arXiv:2011.04555},
  2020.

\bibitem{liu2020joint}
Z.~Liu, Y.~Han, J.~Fan, L.~Zhang, and Y.~Lin, ``Joint optimization of spectrum
  and energy efficiency considering the {C-V2X} security: A deep reinforcement
  learning approach,'' \emph{arXiv preprint arXiv:2003.10620}, 2020.

\bibitem{liang2019spectrum}
L.~Liang, H.~Ye, and G.~Y. Li, ``{Spectrum sharing in vehicular networks based
  on multi-agent reinforcement learning},'' \emph{IEEE Journal on Selected
  Areas in Communications}, vol.~37, no.~10, pp. 2282--2292, October 2019.

\bibitem{li2019multi}
Z.~Li and C.~Guo, ``Multi-agent deep reinforcement learning based spectrum
  allocation for {D2D} underlay communications,'' \emph{IEEE Transactions on
  Vehicular Technology}, vol.~69, no.~2, pp. 1828--1840, February 2019.

\bibitem{feriani2021single}
A.~Feriani and E.~Hossain, ``Single and multi-agent deep reinforcement learning
  for {AI}-enabled wireless networks: A tutorial,'' \emph{IEEE Communications
  Surveys \& Tutorials}, March 2021.

\bibitem{kaul2011AoI}
S.~Kaul, M.~Gruteser, V.~Rai, and J.~Kenney, ``Minimizing age of information in
  vehicular networks,'' in \emph{2011 8th Annual IEEE Communications Society
  Conference on Sensor, Mesh and Ad Hoc Communications and Networks}, June
  2011, pp. 350--358.

\bibitem{ultra_aoi_vehicular}
M.~K. {Abdel-Aziz}, C.~{Liu}, S.~{Samarakoon}, M.~{Bennis}, and W.~{Saad},
  ``Ultra-reliable low-latency vehicular networks: Taming the age of
  information tail,'' in \emph{IEEE Global Communications Conference
  (GLOBECOM)}, December 2018, pp. 1--7.

\bibitem{chen2020age}
X.~Chen, C.~Wu, T.~Chen, H.~Zhang, Z.~Liu, Y.~Zhang, and M.~Bennis, ``Age of
  information aware radio resource management in vehicular networks: A
  proactive deep reinforcement learning perspective,'' \emph{IEEE Transactions
  on Wireless Communications}, vol.~19, no.~4, pp. 2268--2281, April 2020.

\bibitem{sheikh2020multi}
H.~U. Sheikh and L.~Bölöni, ``Multi-agent reinforcement learning for problems
  with combined individual and team reward,'' in \emph{2020 International Joint
  Conference on Neural Networks (IJCNN)}, September 2020, pp. 1--8.

\bibitem{fujimoto2018addressing}
S.~Fujimoto, H.~van Hoof, and D.~Meger, ``Addressing function approximation
  error in actor-critic methods,'' in \emph{Proceedings of the 35th
  International Conference on Machine Learning}, vol.~80.\hskip 1em plus 0.5em
  minus 0.4em\relax PMLR, July 2018, pp. 1587--1596.

\bibitem{van2017hybrid}
H.~Van~Seijen, M.~Fatemi, J.~Romoff, R.~Laroche, T.~Barnes, and J.~Tsang,
  ``Hybrid reward architecture for reinforcement learning,'' in \emph{Annual
  Conference on Neural InformationProcessing Systems (NIPS)}, 2017, p.
  5396–5406.

\bibitem{3gpp_37_885}
``{Study on evaluation methodology of new vehicle-toeverything V2X use cases
  for LTE and NR (Release 16)},'' {3rd Generation Partnership Project (3GPP)},
  Technical Specification (TS) 37.885, June 2018, version 15.3.0.

\bibitem{bultitude20074}
\BIBentryALTinterwordspacing
Y.~d.~J. Bultitude and T.~Rautiainen, ``{IST-4-027756 WINNER II D1. 1.2 V1. 2
  WINNER II Channel Models},'' \emph{EBITG, TUI, UOULU, CU/CRC, NOKIA, Tech.
  Rep}, 2007. [Online]. Available:
  \url{https://www.cept.org/files/8339/winner2%20-%20final%20report.pdf}
\BIBentrySTDinterwordspacing

\bibitem{mccloskey1989catastrophic}
M.~McCloskey and N.~J. Cohen, ``Catastrophic interference in connectionist
  networks: The sequential learning problem,'' in \emph{Psychology of learning
  and motivation}.\hskip 1em plus 0.5em minus 0.4em\relax Elsevier, January
  1989, vol.~24, pp. 109--165.

\end{thebibliography}

\end{document}